\documentclass[final,3p,times,twocolumn,authoryear]{elsarticle}
\usepackage{graphicx}      
\makeatletter
\let\old@ssect\@ssect 
\makeatother

\usepackage{amssymb}
\usepackage{amsfonts}
\usepackage{mathrsfs}
\usepackage{amsmath}
\usepackage{textcomp}
\usepackage{float}
\usepackage[dvipsnames]{xcolor}
\usepackage{multirow}
\usepackage{csquotes}
\usepackage{epstopdf} 
\usepackage{gensymb}
\usepackage{breqn}
\usepackage{siunitx}
\usepackage{hyperref} 
\hypersetup{colorlinks=true,colorlinks,linkcolor={blue},citecolor={blue},urlcolor={red}} 
\newcommand{\rom}[1]{\expandafter{\romannumeral #1\relax}}
\usepackage{enumitem}
\newlist{legal}{enumerate}{10}
\setlist[legal]{label*=\arabic*.}
\makeatletter
\def\endfigure{\end@float}
\def\endtable{\end@float}
\makeatother

\usepackage{subcaption}

\usepackage{booktabs, caption, nicematrix}
\usepackage[thicklines]{cancel}
\usepackage{amsthm}
\newtheoremstyle{boldtheorem}
  {\topsep}   
  {\topsep}   
  {\normalfont}
  {}          
  {\bfseries} 
  {.}         
  {.5em}      
  {}          

\usepackage{tikz}

\theoremstyle{boldtheorem}
\newtheorem{thm}{Theorem}
\newtheorem{rem}{Remark}%
\newtheorem{lem}{Lemma}%
\newtheorem{assum}{Assumption}%

\usepackage{xcolor}

\usepackage[flushleft]{threeparttable}
\usepackage{xcolor}
\usepackage[thinlines]{easytable}

\begin{document}

\begin{frontmatter}




\title{Enhancing Reset Control Phase with Lead Shaping Filters: Applications to Precision Motion Systems}

\author[inst]{Xinxin Zhang}\ead{X.Zhang-15@tudelft.nl} 
\author[inst]{S. Hassan HosseinNia}\ead{S.H.HosseinNiaKani@tudelft.nl (The corresponding author)}
\affiliation[inst]{organization={Department of Precision and Microsystems Engineering (PME), Delft University of Technology},
            addressline={Mekelweg 2}, 
            city={Delft},
            postcode={2628CD}, 
            country={The Netherlands}}

\begin{abstract}
This study presents a shaped reset feedback control strategy to enhance the performance of precision motion systems. The approach utilizes a phase-lead compensator as a shaping filter to tune the phase of reset instants, thereby shaping the nonlinearity in the first-order reset control. {The design achieves either an increased phase margin while maintaining gain properties or improved gain without sacrificing phase margin, compared to reset control without the shaping filter.} Then, frequency-domain design procedures are provided for both Clegg Integrator (CI)-based and First-Order Reset Element (FORE)-based reset control systems. Finally, the effectiveness of the proposed strategy is demonstrated through two experimental case studies on a precision motion stage. In the first case, the shaped reset control leverages phase-lead benefits to achieve zero overshoot in the transient response. In the second case, the shaped reset control strategy enhances the gain advantages of the previous reset element, resulting in improved steady-state performance, including better tracking precision and disturbance rejection, while reducing overshoot for an improved transient response.
\end{abstract}



\begin{keyword}
Shaped first-order reset feedback control \sep Precision motion systems \sep Frequency-domain design \sep Steady-state\sep Transient response
\MSC 93C80 \sep 93C10 \sep 70Q05
\end{keyword}

\end{frontmatter}

\section{Introduction}
\label{sec: intro}
This study focuses on developing reset feedback control strategies to enhance the performance of precision positioning systems. High-precision industries, such as semiconductor manufacturing and robotics, demand systems capable of delivering accurate positioning, effective disturbance and noise rejection, fast response times, stability, and robustness (\cite{schmidt2020design}). To address these requirements, effective control strategies are crucial.

Linear feedback control, particularly the classical Proportional-Integral-Derivative (PID) controller, remains widely used due to its simplicity and effectiveness (\cite{han2009pid}). To meet the demands of industrial precision motion control, the loop-shaping technique is commonly employed in linear control design. This technique focuses on maintaining high gain at low frequencies to ensure effective low-frequency reference tracking and disturbance rejection (\cite{fuller1976feedback}). At the same time, low gain at high frequencies is maintained to reduce sensitivity to high-frequency sensor noise and external disturbances (\cite{schmidt2020design}). Additionally, achieving an appropriate phase margin around the system’s bandwidth is crucial for ensuring stability and a desired transient response (\cite{chang1990gain}), thereby facilitating reliable and smooth operation.

 However, linear controllers face fundamental frequency-domain constraints, such as the waterbed effect and the Bode gain-phase trade-off (\cite{chen2019development}). These limitations restrict their ability to meet the increasingly stringent performance demands of precision motion systems (\cite{saikumar2019constant}). Consequently, advanced control strategies are needed to overcome these trade-offs and achieve superior performance, addressing the evolving demands of precision motion systems.

Nonlinear control strategies, specifically reset feedback control, have emerged as a promising alternative (\cite{banos2012reset}). {Reset control has been applied across diverse industries, including hard-disk-drive systems (\cite{guo2009frequency, guo2010optimal}), wafer stages (\cite{hazeleger2016second, heertjes2016experimental}), undamped second-order plants with time delays (\cite{banos2007design}), minimum-phase relative degree one plants (\cite{zhao2019overcoming}), chemical process control (\cite{carrasco2011reset, banos2012reset}), and mechatronic systems used in this study (\cite{saikumar2019constant, karbasizadeh2022continuous}).} The concept of reset control originated with the Clegg Integrator (CI) in 1958, which resets the integrator's output whenever the input crosses zero. Sinusoidal-Input Describing Function (SIDF) analysis demonstrates that the CI offers a 52° phase lead compared to a linear integrator while maintaining its gain properties (\cite{clegg1958nonlinear, guo2009frequency}). Over time, other reset elements have been introduced to enhance system performance, such as the First-order Reset Element (FORE), Second-order Reset Element (SORE), reset elements with reset bands, and Fractional-order Reset Elements (FrORE), and Constant in Gain Lead in Phase (CgLp) (\cite{ krishnan1974synthesis, horowitz1975non, banos2011limit, hazeleger2016second, saikumar2017generalized, weise2019fractional, saikumar2019constant, weise2020extended}). 

This study focuses on first-order reset controllers, including CI- and FORE-based reset elements such as PI+CI control systems (\cite{banos2007definition}), reset PID controllers (\cite{hosseinnia2013fractional, bisoffi2020stick}), and CgLp controllers. Leveraging their gain and phase advantages, first-order reset controllers have been extensively studied in the literature to enhance transient performance—by reducing overshoot and settling time—and steady-state performance—by improving tracking accuracy and disturbance rejection, particularly in precision motion systems (\cite{zheng2000experimental, heertjes2016experimental, chen2019development, zhao2019overcoming, beerens2019reset, bisoffi2020stick}).

Motivated by the performance of first-order reset controllers, this study aims to further enhance their phase and gain characteristics. Reset control introduces both first-order and high-order harmonics in the frequency domain, and by adjusting reset instants, these harmonics' characteristics can be tailored to improve overall system performance. In closed-loop reset feedback systems, the feedback error signal has traditionally been used as the reset-triggered signal that trigger reset actions. Recent studies have explored alternative reset-triggered signals to tune system performance further. For instance, research in (\cite{karbasizadeh2022band, karbasizadeh2022complex}) developed strategies to modify reset actions to reduce high-order harmonics. However, these techniques focus on reducing high-order harmonics within specific frequency ranges, at the expense of sacrificing the phase and gain characteristics of both first-order and high-order harmonics in other frequency ranges. These limitations restrict the applicability of these methods. In contrast, this work contributes by optimizing the gain and phase of first-order harmonics while preserving the properties of high-order harmonics, thereby improving system performance. The main contributions are as follows:
\begin{itemize}[itemsep=0.01mm, topsep=0.1pt]
    \item First, a linear time-invariant (LTI) phase lead component is proposed as a shaping filter to tune the phase of reset instants, termed shaped reset control. This approach improves the phase-gain margin of the first-order harmonic performance while maintaining similar high-order harmonic characteristics compared to previous reset control strategies. Leveraging the enhanced phase-gain margin, it improves phase lead, resulting in better transient response, or it can be designed to optimize gain properties, leading to superior steady-state performance.
    \item Then, frequency-domain analysis and design procedures are provided for shaped CI- and FORE-based reset elements to achieve phase lead and gain improvements over previous reset control systems. 
    \item Finally, two case studies on a precision motion stage experimentally validate the effectiveness of the shaped reset control strategy. In the first case, the shaped reset PID system introduces phase lead while retaining similar gain properties compared to the reset PID system. This phase lead benefit results in zero-overshoot transient performance, outperforming both the linear PID and reset PID systems. In the second case, the shaped CgLp-PID system is designed to preserve phase margin and high-frequency gain while achieving higher gain at low frequencies and increased bandwidth. These gain enhancements improve tracking precision and disturbance suppression compared to the CgLp-PID and linear PID systems.
\end{itemize}
    
The remainder of the paper is organized into four sections. Section \ref{sec: Background} presents an overview of reset control, covering its definition, stability and convergence conditions, the reset elements employed in this study, and the frequency-domain design objectives for reset control in precision motion systems. Section \ref{sec: contribution} presents the analysis and design procedure of the shaped reset control, highlighting its frequency-domain benefits in terms of phase lead and gain improvements. Section \ref{sec: results} details experimental results conducted on a precision motion stage, validating the effectiveness of the shaped reset control systems compared with linear and reset control systems. Finally, Section \ref{sec: Conclusion} summarizes the main findings and offers suggestions for future research directions.

\section{Preliminaries}
\label{sec: Background}
\subsection{Definition of the Reset Control System}
The reset controller, denoted by \(\mathcal{C}\), is a time-invariant hybrid system (\cite{banos2012reset}). Its state-space representation, with an input signal \(e(t)\), an output signal \(v(t)\), and a state vector \(x_r(t) \in \mathbb{R}^{n_c \times 1}\), is defined as follows:
\begin{equation} 
\label{eq: State-Space eq of reset controller} 
\mathcal{C} = \begin{cases}\dot{x}_r(t) = A_Rx_r(t) + B_Re(t), & t \notin J, \\
x_r(t^+) = A_\rho x_r(t), & t \in J, \\
v(t) = C_Rx_r(t) + D_Re(t).\end{cases}
\end{equation} 
The reset actions of $\mathcal{C}$ in \eqref{eq: State-Space eq of reset controller} are triggered by the zero-crossings of a reset-triggered signal \( e_s(t) \). Consequently, the jump set is defined as \( J := \{ t_i \mid e_s(t_i) = 0, \ i \in \mathbb{Z}^+ \} \), representing an unbounded, monotonically increasing time sequence. For any \( i \in \mathbb{Z}^+ \), it holds that \( t_i < t_{i+1} \) and \(\lim\nolimits_{i \to \infty} t_i \to +\infty \). When \( t \in J \), the jump map of \(\mathcal{C}\) is determined by the matrix \( A_\rho \), given by
\begin{equation}
\label{eq:gamma_def}
		A_\rho =
		\begin{bmatrix}
			\gamma & \\
			& I_{n_c-1}
		\end{bmatrix}, \text{ where }\gamma \in (-1,1) \in\mathbb{R}.
\end{equation}

When \( t \notin J \), the flow dynamics of \(\mathcal{C}\) are defined by the matrices \( A_R \in \mathbb{R}^{n_c \times n_c} \), \( B_R \in \mathbb{R}^{n_c \times 1} \), \( C_R \in \mathbb{R}^{1 \times n_c} \), and \( D_R \in \mathbb{R}^{1 \times 1} \). These matrices characterize the Base-Linear Controller (BLC) \(\mathcal{C}_{bl}\), given by:
\begin{equation}
\label{eq: Cbl}
  \mathcal{C}_{bl}(\omega) = C_R(j\omega I-A_R)^{-1}B_R+D_R, \ j = \sqrt{-1},
\end{equation} where $\omega\in\mathbb{R}^+$ [rad/s] represents the angular frequency. 


Figure \ref{fig: RCS_d_n_r_n_n} depicts the block diagram of a closed-loop reset feedback control system used in this study. This system comprises a reset controller \(\mathcal{C}\) defined in \eqref{eq: State-Space eq of reset controller}, a LTI controller \(\mathcal{C}_\alpha\), and the plant \(\mathcal{P}\). The LTI system \(\mathcal{C}_s\) (where \(\angle \mathcal{C}_s(\omega)\in(-\pi,\pi]\)) is referred to as the \enquote{shaping filter} used to shape the reset actions. Signals \(r\), \(e\), \(e_s\), \(v\), \(u\), \(d\), \(n\), and $y$ denote the reference, error, reset triggered, reset output, control input, process disturbance, sensor noise, and system output signals, respectively. 

\begin{figure}[h]
	\centerline{\includegraphics[width=0.46\textwidth]{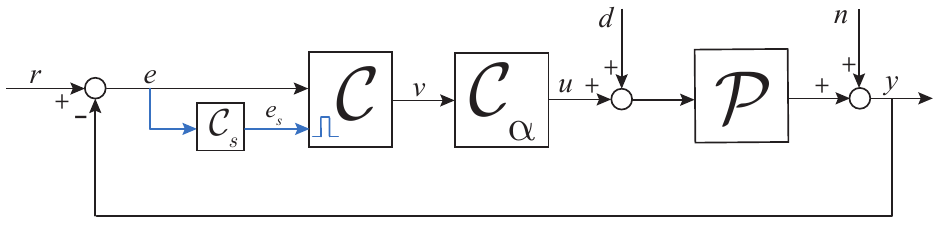}}
	\caption{Block diagram of the closed-loop reset feedback control system, where the blue lines represent the reset action.}
	\label{fig: RCS_d_n_r_n_n}
\end{figure}

\subsection{Stability and Convergence Conditions for Reset Systems}
\label{subsec: Stability and Convergence}
This study focuses on the design of reset control system to enhance system performance. While the stability and convergence of reset systems are not the primary focus, they are essential for the frequency-domain analysis and practical application of such systems. Therefore, the following two assumptions, based on previous literature, outline the necessary conditions to ensure stability and convergence of reset control.

If both $\Delta_i=t_{i+1}-t_i=\delta$ is a constant and $A_\rho \equiv M$ is a constant matrix, then the reset controller \(\mathcal{C}\) \eqref{eq: State-Space eq of reset controller} under an input \( e(t) = |E|\sin(\omega t + \angle E) \), where \( |E| \) and \( \angle E \) denote the magnitude and phase of the signal \( e(t) \) respectively, exhibits a globally asymptotically stable \( 2\pi/\omega \)-periodic solution and converges globally if and only if (\cite{FIRC})
\begin{equation}
\label{eq:open-loop stability}
|\lambda (Me^{A_R}\delta)| <1,\ \ \forall \delta \in \mathbb{R}^+,  
\end{equation}
where $\lambda(\cdot)$ denotes the eigenvalue of $(\cdot)$.    

To ensure the existence of a periodic stable solution for sinusoidal-input reset systems, and thereby enable the sinusoidal-input frequency response analysis of the reset system, the following assumption is introduced:
\begin{assum}
\label{assum: open-loop stable}
The reset system controller \(\mathcal{C}\) satisfies the condition in \eqref{eq:open-loop stability}. The LTI systems \(\mathcal{C}_\alpha\) and \(\mathcal{C}_s\) are Hurwitz.
\label{open-loop stability}   
\end{assum}
Additionally, to ensure the stability and convergence of the closed-loop reset system depicted in Fig. \ref{fig: RCS_d_n_r_n_n}, Assumption \ref{assum: stable} is introduced based on the work (\cite{dastjerdi2022closed}):
\begin{assum}
\label{assum: stable}
The closed-loop reset control system in Fig. \ref{fig: RCS_d_n_r_n_n} is assumed to satisfy the following conditions: the initial condition of the reset controller $\mathcal{C}$ is zero, there are infinitely many reset instants $t_i$ with $\lim\nolimits_{i \to \infty} t_i= +\infty$, the input signals are Bohl functions (\cite{barabanov2001bohl}), and the system meets the $H_\beta$ condition for quadratic stability detailed in (\cite{beker2004fundamental}).
\end{assum}
In practice, Assumption \ref{assum: stable} can be satisfied by employing appropriate design considerations (\cite{banos2012reset, saikumar2021loop}).
\subsection{Reset Elements Used in This Study}
This study focuses on the first-order reset elements, including the CI- and FORE-based reset elements, which are widely applied in the literature and have proven effective for enhancing system performance. The state-space matrices for these reset elements are defined as follows.
\subsubsection{Generalized Clegg Integrator (CI)}
The generalized Clegg Integrator (CI) is characterized by the following matrices:
\begin{equation}
\label{eq: ci_matric}
    A_R=0,  B_R=1, C_R=1, D_R=0, A_\rho=\gamma\in(-1,1).  
\end{equation}
When \(\gamma = 0 \), equation \eqref{eq: ci_matric} characterizes the CI (\cite{clegg1958nonlinear}).

\subsubsection{First-Order Reset Element (FORE)}

The FORE is designed as a Low-Pass Filter (LPF) with a reset mechanism, whose state-space matrices are defined as:
\begin{equation}
\label{eq: fore_matric}
\begin{aligned}
    &A_R = -\omega_{r}, B_R = \omega_{r}, C_R = 1, D_R = 0, \\
    &A_\rho = \gamma \in (-1,1), \quad\text{where } \omega_r \in \mathbb{R^+}.
\end{aligned}
\end{equation}



\subsubsection{Generalized FORE}
In this study, since both the generalized CI in \eqref{eq: ci_matric} and the FORE in \eqref{eq: fore_matric} are first-order reset elements, we define a generalized FORE that collectively describes these elements, with its matrices expressed as:
\begin{equation}
\label{eq: fore_general}
\begin{aligned}
    &A_R = -\omega_{\alpha}, B_R = \omega_{\beta}, C_R = 1, D_R = 0, \\
    &A_\rho = \gamma \in (-1,1), \text{where } \omega_{\alpha} \geq0\in \mathbb{R},\ \omega_{\beta} \in \mathbb{R}^+.
\end{aligned}
\end{equation}
In \eqref{eq: fore_general}, a system with \(\omega_{\alpha} = 0\) and \(\omega_{\beta} =1\) corresponds to the generalized CI in \eqref{eq: ci_matric}, while a system with \(\omega_{\alpha} = \omega_{\beta} > 0\) corresponds to the FORE in \eqref{eq: fore_matric}.

This study aims to design the shaping filter \(\mathcal{C}_s\) to enhance the performance of generalized FORE \eqref{eq: fore_general}-based reset systems, guided by the frequency-domain objectives outlined in the following section.
\subsection{Frequency-Domain Design Objective for Generalized FORE}
In linear systems, the SIDF is commonly employed for analyzing and designing controllers in the frequency domain to meet time-domain performance requirements.

Similarly, for nonlinear systems, where the output contains multiple harmonics, the Higher-Order Sinusoidal Input Describing Function (HOSIDF) is used to perform frequency response analysis (\cite{nuij2006higher}).

Consider a generalized FORE \(\mathcal{C}\) as defined by \eqref{eq: State-Space eq of reset controller} and \eqref{eq: fore_general}, satisfying the condition in \eqref{eq:open-loop stability}, with an input signal \(e(t) = |E|\sin(\omega t)\) and a reset-triggered signal \(e_s(t) = |E|\cdot|\mathcal{C}_s(\omega)|\sin(\omega t + \angle E + \angle \mathcal{C}_s(\omega))\), where \(\angle \mathcal{C}_s(\omega)\in(-\pi,\pi]\). 
The HOSIDF for \(\mathcal{C}\), denoted as \(\mathcal{C}_n(\omega)\), is given by (\cite{Xinxin_zhang_HOSIDF}):
\begin{equation}
\label{eq: fore_cn1}
\resizebox{0.99\columnwidth}{!}{$
    \mathcal{C}_n(\omega)=
    \begin{cases}
 (\Psi(\omega)+ 1) \cdot \omega_{\beta}/(\omega_{\alpha} +j\omega), & \text{ for } n=1,\\ 
 \Psi(\omega)\cdot \omega_{\beta}/(\omega_{\alpha} +jn\omega)\cdot e^{j(n-1)\angle \mathcal{C}_s(\omega)},&\text{ for odd } n>1,\\
     0, &\text{ for even } n\geq2,
    \end{cases}
$}
\end{equation}
where 
\begin{equation}
\label{eq: generalized_fore_alpha}
    \begin{aligned}
    \Lambda(\omega) &= \omega ^2 + \omega_{\alpha}^2,\\
   \Theta(\omega) &= e^{-\pi\omega_{\alpha}/\omega},\\
   \Psi(\omega) &= {2j\omega}\Omega(\omega)\alpha(\omega)/({\pi}\Lambda(\omega)),\\
   \Omega(\omega) &= {(1-\gamma)\cdot(1+\Theta(\omega))}/({1+\gamma\Theta(\omega)}),\\  
   \alpha(\omega)&=e^{j\angle \mathcal{C}_s(\omega)}[\omega \cos(\angle \mathcal{C}_s(\omega))+\omega_{\alpha}\sin(\angle \mathcal{C}_s(\omega))].
    \end{aligned}
\end{equation}
From \eqref{eq: fore_cn1}, the \(n\)th transfer function of the open-loop reset system shown in Fig. \ref{fig: RCS_d_n_r_n_n}, which satisfies Assumption \ref{assum: open-loop stable}, is defined as follows:
\begin{equation}
\label{eq: Ln}
    \mathcal{L}_n(\omega) = \mathcal{C}_n(\omega)\mathcal{C}_\alpha(n\omega)\mathcal{P}(n\omega).
\end{equation}
The bandwidth frequency \( \omega_c \in\mathbb{R}^+\) of a reset control system is defined as the frequency at which the magnitude of the first-order harmonic open-loop transfer function \( \mathcal{L}_1(\omega) \), as given in \eqref{eq: Ln}, reaches 0 dB, mathematically expressed as:
\begin{equation}
\label{eq: wc}
\mathcal{L}_1(\omega_c) = 0 \, \text{dB}.
\end{equation}

In this study, the proposed shaped reset control element is designed to enhance the performance of precision motion systems by satisfying the first-order harmonic \(\mathcal{L}_1(\omega)\) requirements specified in Remark \ref{rem: C1_demand}, while preserving similar high-order harmonics \(\mathcal{L}_n(\omega)\) for \(n > 1\).
\begin{rem}
\label{rem: C1_demand}
Inspired by the loop-shaping technique in linear precision motion control, the design of the first-order harmonic \( \mathcal{L}_1(\omega) \) in \eqref{eq: Ln} for open-loop reset feedback control systems aims to achieve the following key objectives:

(\rom{1}) Ensuring a phase margin of \( \angle \mathcal{L}_1(\omega_c) + 180^\circ \) at the bandwidth frequency \( \omega_c \) defined in \eqref{eq: wc}, to guarantee system stability and optimize transient performance.

(\rom{2}) Maintaining a high gain \( |\mathcal{L}_1(\omega)| \) at frequencies where \( \omega < \omega_c \) to ensure low-frequency reference tracking precision and disturbance rejection.

(\rom{3}) Achieving low gain \( |\mathcal{L}_1(\omega)| \) at frequencies where \( \omega > \omega_c \) to suppress high-frequency noise and improve robustness.
\end{rem}

\section{Frequency-Domain Analysis and Design of the Shaped Reset Feedback Control System}
\label{sec: contribution}
In this section, we first present the phase properties of the generalized FORE derived from its HOSIDF, as detailed in Remark \ref{rem: Cs_Cn} and Remark \ref{eq: phase_CI_BW}. Subsequently, Lemmas \ref{lem: phase of Cs to angle C1} and \ref{rem: alpha_requirement} outline the conditions necessary to enhance the phase margin of the generalized FORE while maintaining similar gain properties. To fulfill these conditions, Theorems \ref{thm: ci_Cs_design} and \ref{thm: fore_Cs_design} establish the requirements for designing the shaping filter \(\mathcal{C}_s\) for CI and FORE elements. Finally, design procedures are provided for the shaped generalized FORE to improve system performance.

\subsection{Frequency-Domain Analysis of Shaping Filters to Achieve Phase Lead in Generalized FROE}

From the HOSIDF expressions for the generalized FORE in \eqref{eq: fore_cn1} and \eqref{eq: generalized_fore_alpha}, two key properties of \(\mathcal{C}_n(\omega)\) are identified. First, Remark \ref{rem: Cs_Cn} highlights the impact of the shaping filter \(\mathcal{C}_s(\omega)\) on the HOSIDF \(\mathcal{C}_n(\omega)\).
\begin{rem}
\label{rem: Cs_Cn}
The phase of the shaping filter, \(\angle \mathcal{C}_s(\omega)\), and the HOSIDF of the generalized FORE, \(\mathcal{C}_n(\omega)\), are related by \(\mathcal{C}_n(\angle \mathcal{C}_s(\omega)) = \mathcal{C}_n(\angle \mathcal{C}_s(\omega) + k\pi)\), where \(k \in \mathbb{Z}\). Furthermore, the magnitude of the shaping filter, \(|\mathcal{C}_s(\omega)|\), has no effect on the HOSIDF.
\end{rem}

The following Remark \ref{eq: phase_CI_BW} derives the phase of the first-order harmonic, \(\angle \mathcal{C}_1(\omega)\), at the bandwidth frequency \(\omega_c\) in the generalized FORE.
\begin{rem}
\label{eq: phase_CI_BW}
From \eqref{eq: fore_cn1} and \eqref{eq: generalized_fore_alpha}, the phase of the first-order harmonic \(\mathcal{C}_1(\omega)\) at the bandwidth frequency \(\omega_c\) is expressed as:
\begin{equation}
\label{eq: angle C1}
\angle \mathcal{C}_1(\omega_c) = 
\begin{cases}
\phi_{\lambda}(\omega_c), &\text{ for } \omega_\alpha = 0,\\
\phi_{\alpha}(\omega_c) -\arctan(\omega_c/\omega_\alpha), &\text{ for } \omega_\alpha > 0.
\end{cases}
\end{equation}
where 
\begin{equation}
\label{eq: angle phi_lambda, phi_alpha}
\resizebox{0.9\columnwidth}{!}{$
\begin{aligned}
 \kappa_\zeta(\omega_c) &= {\omega_c}\cdot\Omega(\omega_c)/( {\pi}\cdot\Lambda(\omega_c)),\\
\phi_{\alpha}(\omega_c) &= \arctan \left(\frac{1}{({\kappa_\gamma(\omega_c)\cdot \kappa_\zeta(\omega_c)})^{-1}-\tan(\angle\mathcal{C}_s(\omega_c))}\right),\\
\phi_{\lambda}(\omega_c) &= \arctan \left( \frac{\sin(2\angle \mathcal{C}_s(\omega_c))-\pi(1+\gamma)/(2(1-\gamma))}{\cos(2\angle \mathcal{C}_s(\omega_c))+1} \right),\\
\kappa_\gamma(\omega_c) &= \omega_c\cdot\cos(2\angle \mathcal{C}_s(\omega_c))+\omega_\alpha\cdot\sin(2\angle \mathcal{C}_s(\omega_c))+\omega_c.
\end{aligned}
$}
\end{equation}
Functions $\Lambda(\omega)$ and $\Omega(\omega)$ are defined in \eqref{eq: generalized_fore_alpha}.
\end{rem}

The performance of the generalized FORE is mainly influenced by three main parameters within the HOSIDF \(\mathcal{C}_n(\omega)\) as defined in \eqref{eq: fore_cn1}, including: (1) the phase of the first-order harmonic at the bandwidth frequency \(\omega_c\): \(\angle \mathcal{C}_1(\omega_c)\) given in \eqref{eq: angle C1}, (2) the magnitude of the first-order harmonic: \(|\mathcal{C}_1(\omega)|\), and (3) the magnitude of the high-order harmonics: \(|\mathcal{C}_n(\omega)|\), for $n>1$. 

In this study, the design of the shaping filter \(\mathcal{C}_s\) aims to provide a phase lead to the first-order harmonic at the bandwidth frequency, \(\angle \mathcal{C}_1(\omega_c)\) as defined in \eqref{eq: angle C1}, while preserving similar gain characteristics \(|\mathcal{C}_n(\omega)|\) compared to the system without the shaping filter (i.e., \(\mathcal{C}_s = 1\)). To achieve this, Lemma \ref{lem: phase of Cs to angle C1} specifies the necessary conditions for the shaping filter to effectively provide the phase lead advantage.
\begin{lem}
\label{lem: phase of Cs to angle C1}
The phase of the first-order harmonic in the generalized FORE at the bandwidth frequency \(\omega_c\), represented as \(\angle \mathcal{C}_1(\omega_c) \in (-\pi, \pi]\), is larger than that of the system without the shaping filter (i.e., \(\mathcal{C}_s = 1\)) if the phase of the shaping filter satisfies the following conditions:
\begin{equation}
\label{eq: CS_Phase_lead_2cond}
    \begin{cases}
     \angle \mathcal{C}_s(\omega_c)\in\left(k\pi,\ \frac{\pi}{2}-\arctan \bigg(\frac{\pi(1+\gamma)}{4(1-\gamma)}\bigg)+k\pi\right), &\text{ for } \omega_\alpha = 0,\\
     \angle \mathcal{C}_s(\omega_c) \in\left(k\pi,\ \frac{\pi}{2}-\arctan \bigg(\frac{\omega_c}{\omega_\alpha}\bigg)+k\pi\right), &\text{ for } \omega_\alpha > 0,
    \end{cases}
\end{equation}
where $k=-1,0$.
\end{lem}
\begin{proof}
The proof is provided in \ref{pf: lemma: angle cs and C1}.
\end{proof}

Lemma \ref{lem: phase of Cs to angle C1} outlines the conditions required for \( \angle \mathcal{C}_s(\omega_c) \) to achieve a phase lead. However, from \eqref{eq: fore_cn1}, altering \( \mathcal{C}_s(\omega) \) modifies the gain properties of \( |\mathcal{C}_n(\omega)| \). To ensure a fair comparison, it is essential to limit these gain variations, which can be achieved by adhering to the constraints in Lemma \ref{rem: alpha_requirement}.
\begin{lem}
\label{rem: alpha_requirement}
To limit the gain variation of \(|\mathcal{C}_n(\omega)|\) in the generalized FORE with a shaping filter \(\mathcal{C}_s \neq 1\), compared to the system where \(\mathcal{C}_s = 1\), the following condition must be satisfied:
    \begin{equation}
       \kappa_\alpha(\omega) \in (1 - \sigma, \ 1 + \sigma),\text{ for } \omega\neq\omega_c,
    \end{equation}  
    where $\sigma \in (0, 1) \subset \mathbb{R}$, and
    \begin{equation}
    \label{eq: kappa_alpha}
    \begin{aligned}
     \kappa_\alpha(\omega) &= |\cos(\angle \mathcal{C}_s(\omega))+\sin(\angle \mathcal{C}_s(\omega))\cdot {\omega_\alpha}/{\omega}|.
    \end{aligned}
    \end{equation}
\end{lem}
\begin{proof}
The proof is provided in \ref{proof for lema_alpha_requirement}.
\end{proof}


In practice, the value of \(\sigma \in (0, 1)\) should be kept small. Specifically, when \(\sigma = 0\), the gain properties of the generalized FORE remain unchanged. By adhering to the constraints in Lemma \ref{rem: alpha_requirement} and choosing an appropriate \(\sigma\), the gain changes can be effectively restricted, ensuring similar gain properties. The selection of \(\sigma\) depends on the system's gain requirements, as demonstrated in the case studies in Section \ref{sec: results}. 

To illustrate the effects of \(\sigma\), we examine the CI with a shaping filter that satisfies the constraints in Lemmas \ref{lem: phase of Cs to angle C1} and \ref{rem: alpha_requirement}, referred to as the shaped CI. Figure \ref{fig: sigma_alpha_changes_C1_C3} presents the magnitude \(|\mathcal{C}_1(\omega)|\) and phase \(\angle \mathcal{C}_1(\omega)\) of the first-order harmonic, along with the magnitude \(|\mathcal{C}_3(\omega)|\) of the third-order harmonic, for both the CI and the shaped CI with \(\gamma = 0\). The analysis considers \(\sigma = 0.01, 0.05, 0.1, 0.2\). 

For clarity, higher-order harmonics $|\mathcal{C}_n(\omega)|$ for \(n>3\) are omitted, as they exhibit the same trend as $|\mathcal{C}_3(\omega)|$ but with smaller magnitudes and minimal variations. Additionally, the shaping filters used in this example, while selected to satisfy Lemmas \ref{lem: phase of Cs to angle C1} and \ref{rem: alpha_requirement}, are not the only possible options. The design of \(\mathcal{C}_s\) will be further discussed in subsequent sections.

The results in Fig. \ref{fig: sigma_alpha_changes_C1_C3} demonstrate a distinct phase lead in $\angle \mathcal{C}_1(\omega)$ with minimal variations in $|\mathcal{C}_n(\omega)|$ for \(n = 1, 3\). Specifically, for \(\sigma = 0.1\), the phase lead at 100 Hz is 12.6 degrees, while the changes in \(|\mathcal{C}_1(\omega)|\) and \(|\mathcal{C}_3(\omega)|\) are negligible. The minimal effects of these small changes will be further shown in the case studies presented in Section \ref{sec: results}. 


\begin{figure}[htp]
	\centerline{\includegraphics[width=0.95\columnwidth]{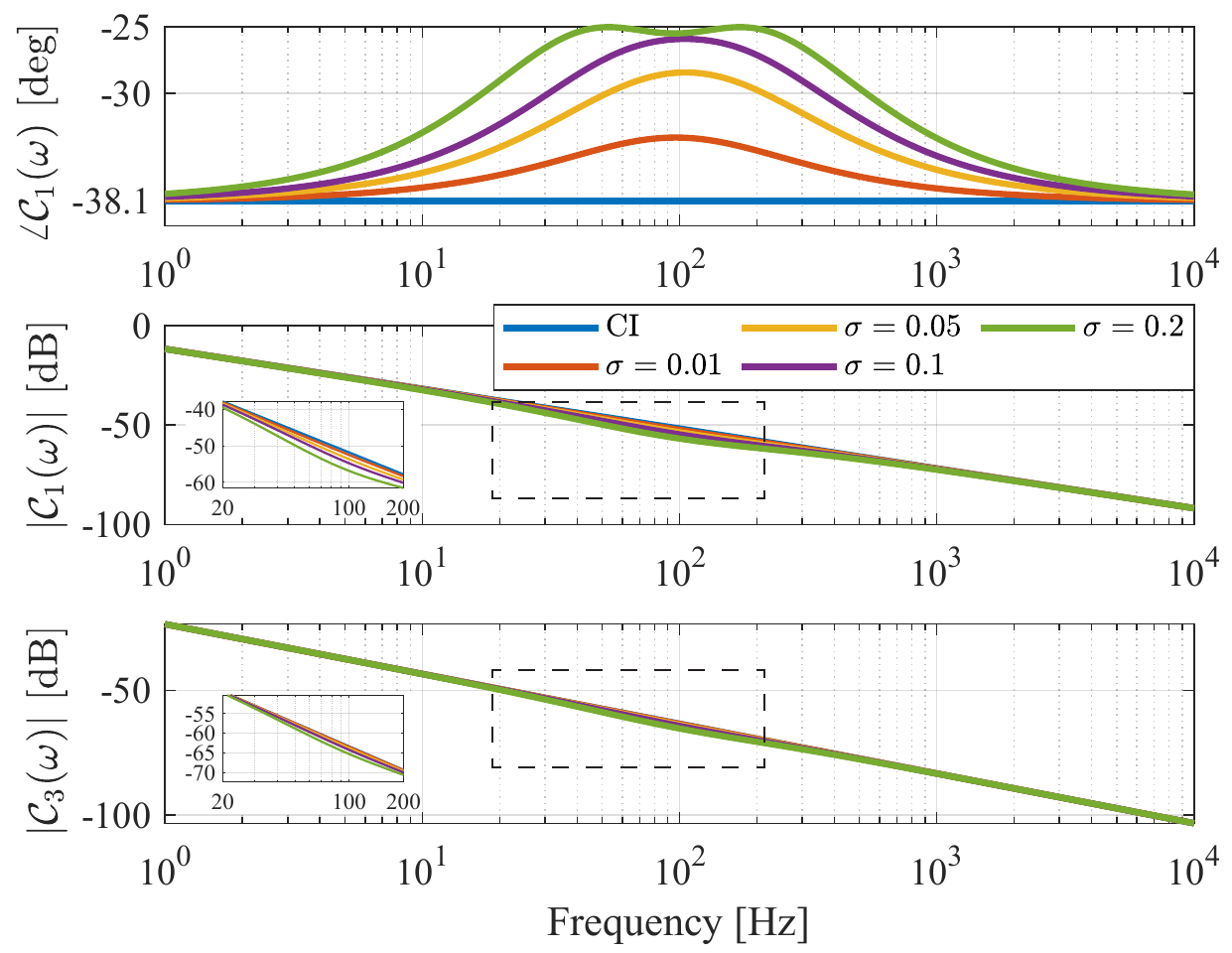}}
	\caption{The magnitudes \(|\mathcal{C}_1(\omega)|\) and phases \(\angle \mathcal{C}_1(\omega)\) of the first-order harmonic, along with the magnitude \(|\mathcal{C}_3(\omega)|\)of the third-order harmonic, for both the CI and the shaped CI with \(\gamma = 0\) considering \(\sigma = 0.01, 0.05, 0.1, 0.2\).}
	\label{fig: sigma_alpha_changes_C1_C3}
\end{figure}

To summarize, Lemmas \ref{lem: phase of Cs to angle C1} and \ref{rem: alpha_requirement} outline the conditions for enhancing the phase margin of the generalized FORE while preserving similar gain benefits. To simultaneously meet these requirements, Theorems \ref{thm: ci_Cs_design} and \ref{thm: fore_Cs_design} specify the conditions for \(\mathcal{C}_s(\omega)\) in the generalized FORE, as defined in \eqref{eq: fore_general}, for cases where \(\omega_\alpha = 0\) (generalized CI) and \(\omega_\alpha > 0\) (FORE), respectively.

\begin{thm}
\label{thm: ci_Cs_design}
In the generalized CI defined in \eqref{eq: ci_matric}, to achieve phase lead while maintaining similar gain properties compared to the system with \(\mathcal{C}_s = 1\), the shaping filter \(\mathcal{C}_s\), where \(\angle \mathcal{C}_s(\omega) \in (-\pi, \pi]\), needs to satisfy the following conditions:
\begin{equation}
\label{eq: cs_ineq_CI}
\begin{cases}
\angle \mathcal{C}_s(\omega_c) \in\left(k\pi,\ \frac{\pi}{2}-\arctan \bigg(\frac{\pi(1+\gamma)}{4(1-\gamma)}\bigg)+k\pi\right),&\text{ for } \omega=\omega_c,\\
\angle \mathcal{C}_s(\omega) \in \{\eta_1\cup\eta_2\cup\eta_3\}, &\text{ for } \omega\neq\omega_c,
\end{cases}
\end{equation}
where $k=-1,0,$ and 
\begin{equation}
\label{eq: eta123}
\begin{aligned}
\eta_1 &= (-\arccos(1-\sigma),\arccos(1-\sigma)),\\
\eta_2 &= (\arccos(-1+\sigma),\pi],\\
\eta_3 &= [-\pi, -\arccos(-1+\sigma)),\ \sigma\in(0,1)\subset\mathbb{R}.
\end{aligned}    
\end{equation}
\end{thm}
{The ranges of \( \eta_1 \), \( \eta_2 \), and \( \eta_3 \) are visualized in Fig. \ref{fig: Desired_Cs_ci_final}.} 
\begin{proof}
The proof is provided in \ref{proof for theorem ci}.
\end{proof}

From \eqref{eq: eta123}, we have 
\begin{equation}
    \eta_1 = \{\eta_2-\pi\} \cup \{\eta_3+\pi\}.
\end{equation}
Since the effects of the shaping filter \(\mathcal{C}_s(\omega)\) on the HOSIDF of the generalized FORE are \(\pi\)-periodic, as noted in Remark \ref{rem: Cs_Cn}, positioning \(\angle \mathcal{C}_s(\omega)\) within \(\eta_2 \cup \eta_3\) can be effectively achieved by positioning it within \(\eta_1\). For reference, we plot a desired curve for \(\angle \mathcal{C}_s(\omega)\) within \(\eta_1\) for \(\omega \neq \omega_c\), while \(\angle \mathcal{C}_s(\omega_c)\) satisfies the constraint outlined in Theorem \ref{thm: ci_Cs_design}. However, the choice of \(\angle \mathcal{C}_s(\omega)\) is not unique; other curves for \(\angle \mathcal{C}_s(\omega)\) that remain within the specified bounds can also achieve phase lead and preserve similar gain.

\begin{figure*}[htp]
	\centerline{\includegraphics[width=1.5\columnwidth]{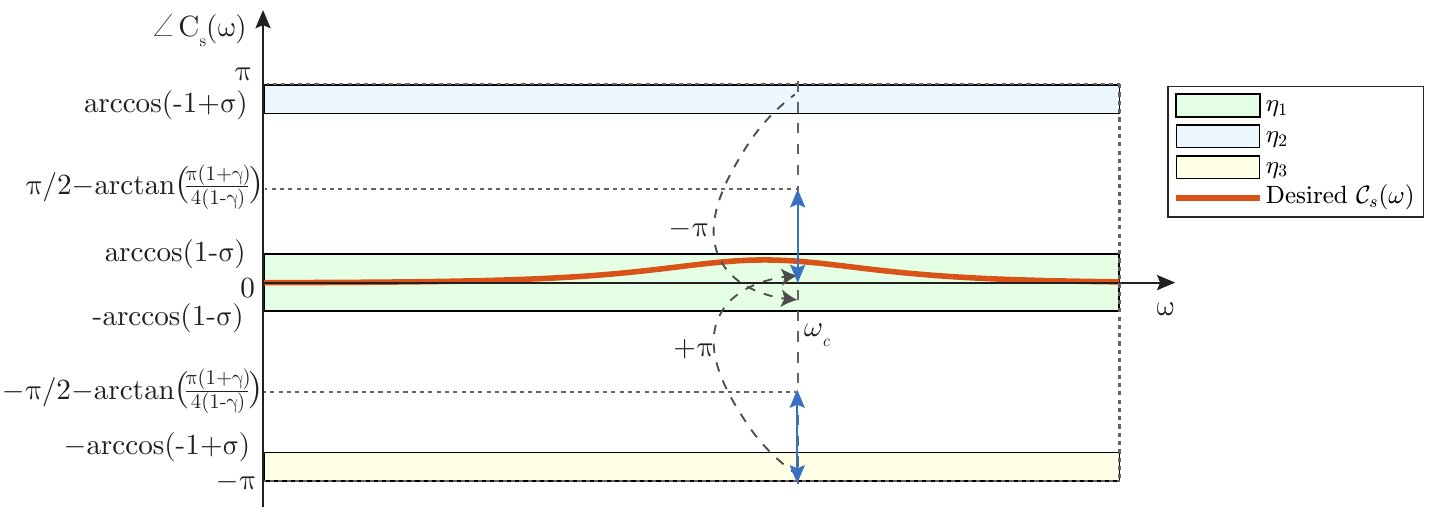}}
	\caption{The three bounds, \(\eta_1\)(\colorbox{green!15}{ }), \(\eta_2\)(\colorbox{yellow!20}{ }), and \(\eta_3\)(\colorbox{RoyalBlue!15}{ }), for \(\angle \mathcal{C}_s(\omega)\) are depicted as shaded regions. The constraint on \(\angle \mathcal{C}_s(\omega)\) at the bandwidth frequency \(\omega_c\) is highlighted with blue double arrows (\textbf{${\leftrightarrow}$}). The desired curve of \(\angle \mathcal{C}_s(\omega)\) for the generalized CI is shown in red, adhering to the constraints.}
	\label{fig: Desired_Cs_ci_final}
\end{figure*}

\begin{thm}
\label{thm: fore_Cs_design}
In the FORE defined in \eqref{eq: fore_matric}, to achieve phase lead while maintaining similar gain properties compared to the system with \(\mathcal{C}_s = 1\), the shaping filter \(\mathcal{C}_s\), where \(\angle \mathcal{C}_s(\omega) \in (-\pi, \pi]\), needs to satisfy the following conditions:
\begin{equation}
\label{eq: cs_ineq_fore}
\begin{cases}
\angle \mathcal{C}_s(\omega_c) \in (k\pi,\frac{\pi}{2}-\arctan (\frac{\omega_c}{\omega_\alpha})+k\pi) , &\text{ for } \omega=\omega_c,\\
\angle \mathcal{C}_s(\omega)\in \{\beta_1\cup\beta_2\cup\beta_3\cup\beta_4\}, &\text{ for } \omega\neq\omega_c,
\end{cases}
\end{equation}
where $k=-1,0,$ and 
\begin{equation}
\label{eq:beta123}
\begin{aligned}
\beta_1 &=  (\arctan \theta_\alpha - \arccos  (\theta_\gamma), \arctan \theta_\alpha - \arccos (\theta_\eta)),\\
\beta_2 &=  (\arctan \theta_\alpha - \arccos (-\theta_\eta), \arctan \theta_\alpha - \arccos (-\theta_\gamma)),\\
\beta_3 &=  \beta_1+\pi,\\
\beta_4 &=  \beta_2+\pi,\\
\theta_\alpha &= \frac{\omega_\alpha}{\omega}, \\
\theta_\gamma &= \frac{1 - \sigma}{\sqrt{1 + \theta_\alpha^2}} ,\ \theta_\eta =  \frac{1 + \sigma}{\sqrt{1 + \theta_\alpha^2}},\ \sigma\in(0,1)\subset\mathbb{R}.
\end{aligned}
\end{equation}
Note that the value of \(\arccos(x)\) is defined within the interval \([0, \pi]\). 
{The ranges of \( \beta_1 \), \( \beta_2 \), \( \beta_3 \), and \( \beta_4 \) are visualized in Fig. \ref{fig: Desired_Cs}.} 
\end{thm}
\begin{proof}
The proof is provided in \ref{proof for theorem fore}.
\end{proof}

Similar to Fig. \ref{fig: Desired_Cs_ci_final}, a desired curve for \(\angle \mathcal{C}_s(\omega)\) is plotted within the bounds of \(\beta_1 \cup \beta_4\) for \(\omega \neq \omega_c\), while \(\angle \mathcal{C}_s(\omega_c)\) is constrained by the condition outlined in Theorem \ref{thm: fore_Cs_design}.
\begin{figure*}[htp]
	\centerline{\includegraphics[width=1.5\columnwidth]{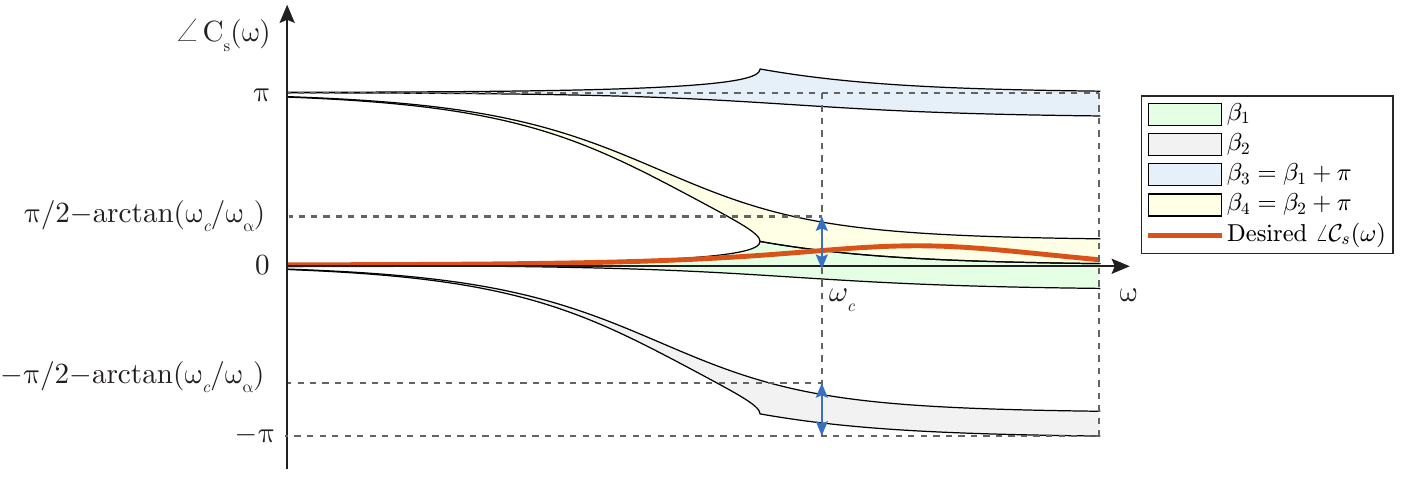}}
	\caption{The four bounds, \(\beta_1\)(\colorbox{green!15}{ }), \(\beta_2\)(\colorbox{gray!30}{ }), \(\beta_3\)(\colorbox{RoyalBlue!20}{ }), and \(\beta_4\)(\colorbox{yellow!25}{ }), for \(\angle \mathcal{C}_s(\omega)\) are depicted as shaded regions. The constraint on \(\angle \mathcal{C}_s(\omega)\) at the bandwidth frequency \(\omega_c\) is highlighted with blue double arrows (\textbf{${\leftrightarrow}$}). The desired curve of \(\angle \mathcal{C}_s(\omega)\) for the FORE is shown in red, adhering to the constraints.}
	\label{fig: Desired_Cs}
\end{figure*}

\subsection{Frequency-Domain Design of Shaped Generalized FROE to Enhance System Performance}
While various shaping filters \(\mathcal{C}_s\) satisfying the constraints in Theorems \ref{thm: ci_Cs_design} and \ref{thm: fore_Cs_design} can be selected to achieve phase lead while maintaining similar gain properties, this study adopts a derivative element:  
\begin{equation}
\mathcal{C}_s(s) = \frac{s/\omega_\zeta+1}{s/\omega_\eta+1}, \quad \text{where } \omega_\zeta,\omega_\eta \in \mathbb{R}^+,
\end{equation}  
which aligns with the desired phase curve shapes of \(\angle \mathcal{C}_s(\omega)\) illustrated in both Fig. \ref{fig: Desired_Cs_ci_final} for the generalized FORE with \(\omega_\alpha = 0\) and Fig. \ref{fig: Desired_Cs} for the generalized FORE with \(\omega_\alpha > 0\), respectively.

However, implementing a single derivative element between the error signal \(e(t)\) and the reset-triggered signal \(e_s(t)\) can amplify high-frequency harmonics for frequencies \(\omega > \omega_\eta\) in \(e_s(t)\). In practical scenarios, especially when high-frequency noise from sensors or external interference is present, this amplification can increase the system's sensitivity to such noise, potentially compromising its steady-state performance. 

To mitigate this issue, a low-pass filter \(\frac{1}{s/\omega_\psi + 1}\) is needed to filter out high-frequency harmonics in the reset-triggered signal \(e_s(t)\). {The design of \(\omega_\psi\) ensures that \(|\mathcal{C}_s(\omega)| < \delta_n\) for \(\omega > \omega_c\), where \(\omega_c\) is the bandwidth frequency and \(\delta_n \in (1,2) \subset \mathbb{R}\). The value of \(\delta_n\) is selected based on the noise level of the setup. In this manuscript, we set \(\delta_n = 1.5\) accordingly.} Therefore, the transfer function of the shaping filter \(\mathcal{C}_s(s)\) is designed as: 
\begin{equation}
\label{eq: cs_def}
\mathcal{C}_s(s) = \frac{s/\omega_\zeta+1}{s/\omega_\eta+1}\cdot\frac{1}{s/\omega_\psi+1},
\end{equation}
where 
$\omega_\zeta, \omega_\eta \in \mathbb{R}^+, \text{ and } \omega_\psi\in \mathbb{R}^+ >\omega_\eta$.

Note that while using a second-order or higher-order phase-lead element as the shaping filter can also provide phase lead, but it may exacerbate the issue of high-frequency noise amplification in the reset-triggered signal \( e_s(t) \), making the system less robust to practical noise. The feasibility of using a higher-order lead element is outside the scope of this study and requires further investigation.

The reset control system with a shaping filter, defined in \eqref{eq: cs_def} and satisfying the conditions specified in Theorems \ref{thm: ci_Cs_design} and \ref{thm: fore_Cs_design}, is referred to as the shaped reset control system in this study. The phase lead at the bandwidth frequency \(\omega_c\), provided by the shaping filter $\mathcal{C}_s$, is calculated as described in Remark \ref{eq:phi_lead_cal}.
\begin{rem}
\label{eq:phi_lead_cal}
The phase lead of the shaped generalized FORE with the shaping filter \( \mathcal{C}_s(s) \neq 1 \) compared to the generalized FORE where \( \mathcal{C}_s(s) = 1 \) is given by:
\begin{equation}
\label{eq:phi_lead}
 \phi_{\text{lead}} = \angle \mathcal{C}_1(\omega_c) - \angle \mathcal{C}_1^0(\omega_c),   
\end{equation}
where \(\angle \mathcal{C}_1(\omega_c)\) represents the phase of the shaped generalized FORE, which can be calculated using \eqref{eq: angle C1}, and \(\angle \mathcal{C}_1^0(\omega_c)\) represents the phase of the generalized FORE with \(\mathcal{C}_s = 1\), as given by:
\begin{equation}
\label{eq: angle C1_0}
\begin{aligned}
 &\angle \mathcal{C}_1^0(\omega_c)= \\
&\begin{cases}
\arctan \left( \frac{-\pi(1+\gamma)}{4(1-\gamma)}\right), &\text{ for } \omega_\alpha = 0,\\
\arctan \left( 2\omega_c\cdot\kappa_\zeta(\omega_c)\right) -\arctan\left(\frac{\omega_c}{\omega_\alpha}\right), &\text{ for } \omega_\alpha > 0,
\end{cases}
\end{aligned}
\end{equation}  
with $\kappa_\zeta(\omega_c)$ given in \eqref{eq: angle phi_lambda, phi_alpha}.
\end{rem}
MATLAB code for calculating the phase lead \(\phi_{\text{lead}}\) in \eqref{eq:phi_lead} is available at this \href{https://github.com/XZ-TUD/Code_Phase_Lead.git}{link} to facilitate ease of use for readers. Next, Remark \ref{rem:max_phi} presents the maximum phase lead that can be achieved by the shaping filter under the constraints specified in Theorems \ref{thm: ci_Cs_design} and \ref{thm: fore_Cs_design}.
\begin{rem}
\label{rem:max_phi}
From Lemma \ref{lem: phase of Cs to angle C1}, the maximum phase of shaping filter \(\angle \mathcal{C}_s(\omega_c) \in(-\pi,\pi]\) is given by
\begin{equation}
\label{eq:max_cs_angle_bw}
\max \angle \mathcal{C}_s(\omega_c) = 
\begin{cases}
  \frac{\pi}{2}-\arctan\bigg(\frac{\pi(1+\gamma)}{4(1-\gamma)}\bigg), & \text{for } \omega_\alpha = 0,\\
  \frac{\pi}{2}-\arctan\bigg(\frac{\omega_c}{\omega_\alpha}\bigg), & \text{for } \omega_\alpha > 0.  
\end{cases}
\end{equation}
By substituting \(\max \angle \mathcal{C}_s(\omega_c)\) from \eqref{eq:max_cs_angle_bw} into \eqref{eq:phi_lead} and \eqref{eq: angle C1_0}, the maximum phase lead, denoted as \(\max \phi_{\text{lead}}\), of the shaped generalized FORE (where \(\mathcal{C}_s \neq 1\)) compared to the generalized FORE without the shaping filter (where \(\mathcal{C}_s = 1\)) can be determined.
\end{rem}

Finally, summarizing the constraints in Theorems \ref{thm: ci_Cs_design} and \ref{thm: fore_Cs_design}, along with conclusions in Remarks \ref{eq:phi_lead_cal} and \ref{rem:max_phi}, the design procedure for the shaping filter \(\angle \mathcal{C}_s(s)\) in the shaped generalized FORE-based reset control system, aimed at achieving a phase lead \(\phi_d \in(0, \max \phi_{\text{lead}}]\) compared to the generalized FORE-based reset control system with \(\mathcal{C}_s=1\), is outlined as follows:
\begin{enumerate}[label = (\roman*)]
\item Design a generalized FORE-based reset control system without the shaping filter (i.e., $\mathcal{C}_s=1$) and set the bandwidth frequency \(\omega_c\).

\item Apply a shaping filter $\mathcal{C}_s$ as defined in \eqref{eq: cs_def}.

\item Choose \(\sigma \in (0, 1)\). Next, tune \(\omega_\zeta\), \(\omega_\eta\), and \(\omega_\psi\) in $\mathcal{C}_s(\omega)$ to satisfy the conditions specified in Theorem \ref{thm: ci_Cs_design} if \(\omega_\alpha = 0\), and in Theorem \ref{thm: fore_Cs_design} if \(\omega_\alpha > 0\).

\item Calculate the phase lead \(\phi_{\text{lead}}\) provided by the shaping filter using \eqref{eq:phi_lead}. If \(\phi_{\text{lead}} < \phi_d\), decrease \(\omega_\zeta\) or increase \(\omega_\eta\), and repeat from step (\rom{3}) until \(\phi_{\text{lead}} = \phi_d\).
\end{enumerate}

If the system requirements prioritize gain improvement over phase margin enhancement, the design procedure for shaping the filter \(\mathcal{C}_s(s)\) involves first following the above steps to achieve phase lead, and then transferring this phase lead benefit to gain improvement by relaxing the gain constraint in Lemma \ref{rem: alpha_requirement} for frequencies \(\omega \neq \omega_c\). The design procedure to obtain gain benefits while maintaining phase margin compared to a generalized FORE-based reset control system with $\mathcal{C}_s=1$ is outlined as follows:
\begin{enumerate} [label = (\roman*)]
\item Design a shaped generalized FORE-based reset control system to provide a phase lead \(\phi_{\text{lead}}\).
\item Gradually adjust parameters such as \(\omega_\alpha\) and \(\gamma\) to increase the first-order harmonic gain \(|\mathcal{C}_1(\omega)|\) at frequencies below \(\omega_c\) or reducing gain at higher frequencies. As gain benefits increase, the phase lead \(\phi_{\text{lead}}\) diminishes; tuning continues until \(\phi_{\text{lead}} = 0\), where the shaped generalized FORE maintains phase margin while maximizing gain benefits. 
\end{enumerate}

Note that for the generalized FORE with \(\omega_\alpha > 0\), both \(\omega_\alpha\) and \(\gamma\) offer flexibility in tuning; in contrast, systems with \(\omega_\alpha = 0\) rely solely on \(\gamma\). Therefore, the FORE-based control systems with \(\omega_\alpha > 0\) are preferable for providing enhanced gain benefits due to their greater tuning flexibility.

In Section \ref{sec: results}, two case studies are presented to demonstrate the design procedure of shaped generalized FORE control systems, aiming to achieve phase and gain benefits, respectively.
\section{Illustrative Case Studies}
\label{sec: results}
In this section, the experimental setup-a precision positioning stage-is first introduced. Two case studies are then conducted on this stage to demonstrate the enhanced performance of the shaped generalized FORE-based reset control system:
\begin{itemize}
\itemsep0em 
    \item Case Study 1 uses a reset PID controller to showcase the phase lead advantages provided by the shaped reset control. 
    \item Case Study 2 employs a CgLp-PID control system to emphasize the gain benefits, particularly achieving enhanced low-frequency gain. 
\end{itemize}
Note that these cases may not represent the optimized designs; and the aim of these cases is to illustrate how the shaped reset control can offer improvements over previous reset control systems under a fair comparison framework. {Additionally, the performance of the shaped reset control is not confined to the specific stages used in the case studies or the results presented in this section. Depending on different and specific system requirements, shaped reset control can be designed using the methodologies outlined in Theorem \ref{thm: ci_Cs_design} and Theorem \ref{thm: fore_Cs_design}.} In both cases, the systems are tested to be stabile and convergent. 
\label{sec: results1}
\subsection{Precision Positioning Setup}
Figure \ref{fig: Spider_stage} illustrates the precision positioning setup utilized in this study. The system consists of a 3 Degree-Of-Freedom (DOF) stage mounted on a vibration isolation platform to minimize the impact of environmental disturbances. Control algorithms are implemented on an NI CompactRIO system equipped with FPGA modules, operating at a sampling frequency of 10 kHz. The voice coil actuation system is powered by a linear current source amplifier {(with a power supply limit of 10 V)}, while position feedback is acquired using a Mercury M2000 linear encoder (referred to as ``Enc") with a resolution of 100 nm.

The 3 DOF precision positioning stage consists of three masses, \(M_1\), \(M_2\), and \(M_3\), which are connected to the base mass \(M_c\) via dual leaf flexures. Each of these masses is associated with an actuator: \(A_1\), \(A_2\), and \(A_3\), respectively. In this study, the collocated system comprising actuator \(A_1\) and mass \(M_1\) is utilized for control implementation and performance evaluation. Figure \ref{fig1: Bode2 of Spyder} presents the measured Frequency Response Function (FRF) of the system. To facilitate feedback control design, the system's transfer function is approximated as an LTI model using Matlab's system identification toolbox, which simplifies the system to a single-eigenmode mass-spring-damper configuration:
\begin{equation}
\label{eq:P(s)}
\mathcal{P}(s) = \frac{6.615\times 10^5}{83.57s^2+279.4s+5.837\times 10^5}.
\end{equation}
\begin{figure}[htp]
	\centerline{\includegraphics[width=0.9\columnwidth]{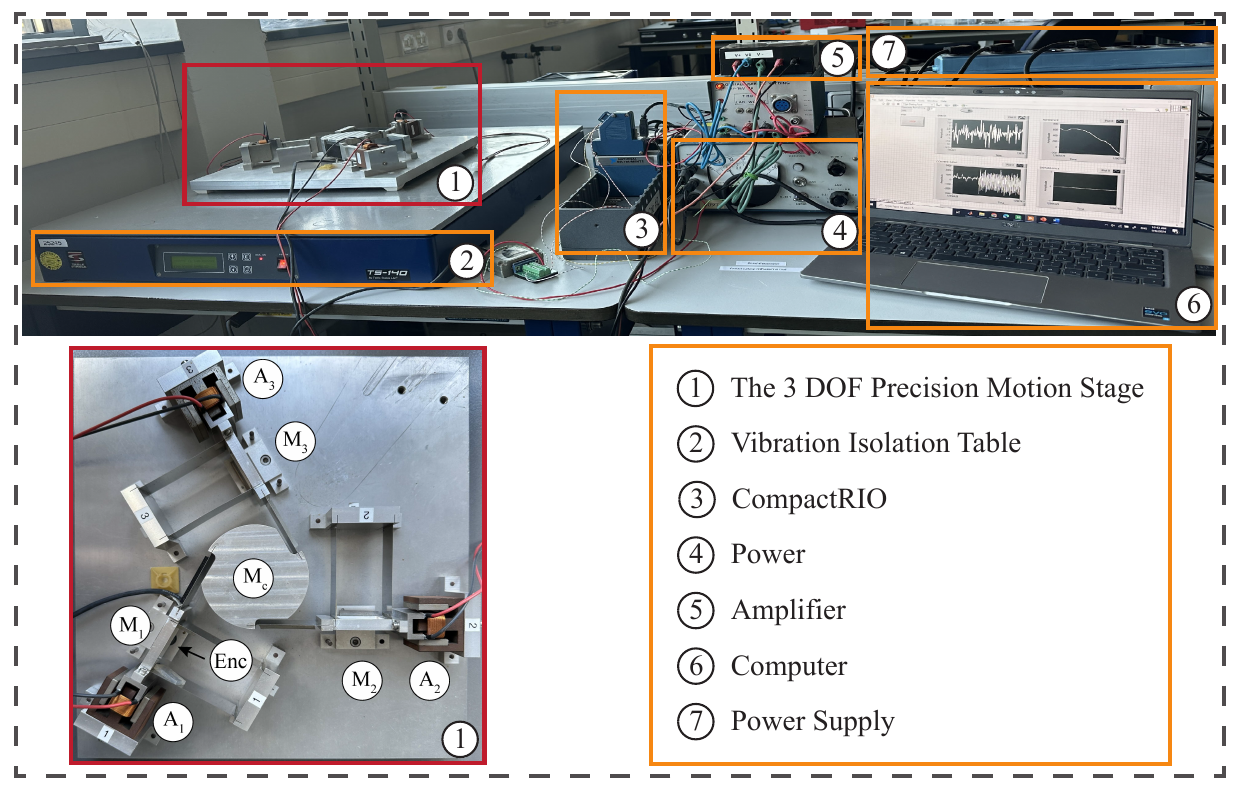}}
	\caption{Experimental precision positioning setup.}
	\label{fig: Spider_stage}
\end{figure}
\begin{figure}[htp]
	\centerline{\includegraphics[width=0.9\columnwidth]{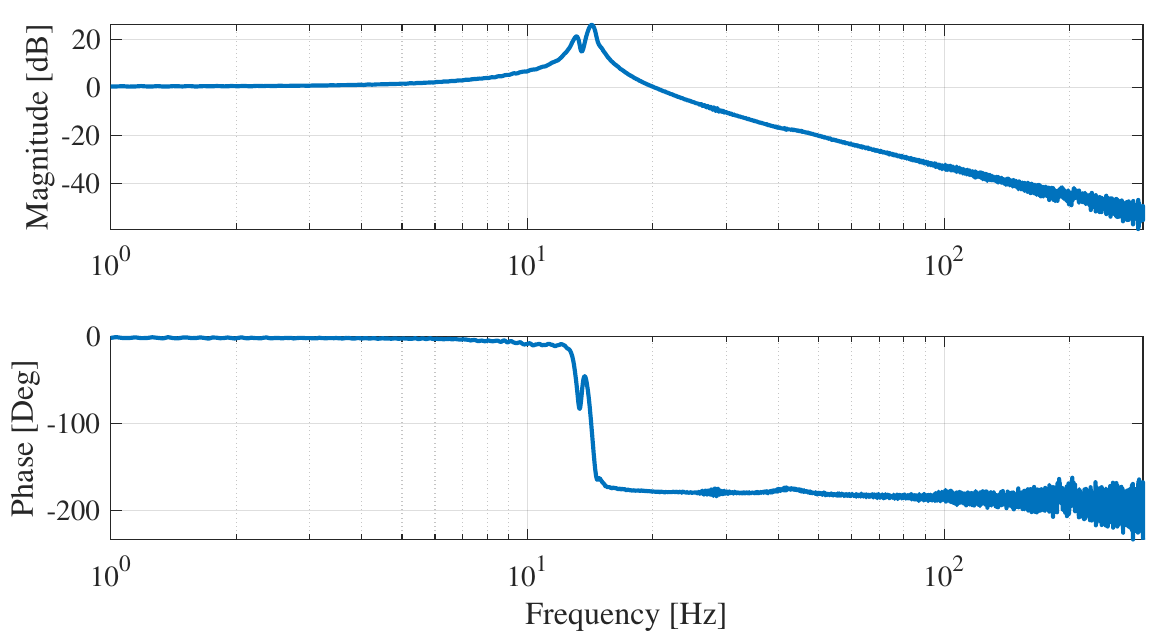}}
	\caption{Measured FRF data from actuator \(A_1\) to attached mass \(M_1\) of the precision positioning stage.}
	\label{fig1: Bode2 of Spyder}
\end{figure}
\subsection{Case Study 1: Phase Lead Benefit of Shaped Reset Control Resulting in Transient Performance Improvement}
\label{sec: case1}
In Case Study 1, a reset PID control system is designed to showcase the phase lead benefit of shaped reset control within the framework of the generalized FORE-based reset control when \(\omega_\alpha = 0\). This design is informed by Theorem \ref{thm: ci_Cs_design}. The following content illustrates the design and comparison process.

By replacing the Proportional Integrator (PI) with the Proportional Clegg Integrator (PCI) in the PID control system, a Proportional Clegg Integrator Derivative (PCID) system is built. {However, the closed-loop PCID system tends to exhibit a limit cycle behavior (\cite{hosseinnia2013fractional}). To mitigate this issue, one approach is to incorporate an additional integrator, resulting in the PCI-PID system, whose block diagram is shown in Fig. \ref{fig: PCI_PID_Structure}.}
\begin{figure}[htp]
	\centerline{\includegraphics[width=0.9\columnwidth]{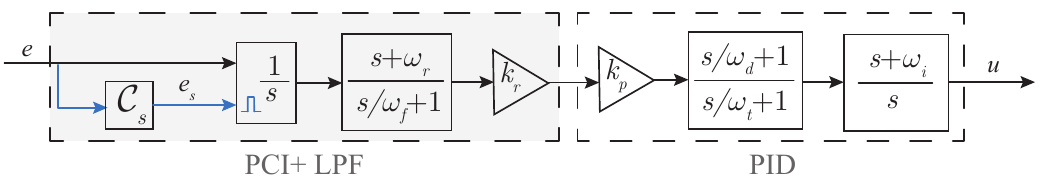}}
	\caption{Block diagram of the PCI-PID control system.}
	\label{fig: PCI_PID_Structure}
\end{figure}

{By designing the PCI reset elements shown within the gray block in Fig. \ref{fig: PCI_PID_Structure}, the PCI-PID system can leverage gain benefits while maintaining the same phase characteristics as its base linear system}, the PI\(^2\)D system, as given by:
\begin{equation}
    \text{PI$^2$D} = k_p \cdot \bigg(\frac{s + \omega_i}{s} \bigg)^2 \cdot \frac{s/\omega_d + 1}{s/\omega_t + 1} \cdot \frac{1}{s/\omega_f + 1}.
\end{equation}

The design parameters of the PCI-PID control system are: \(\omega_r = 1.6 \times 10^3\) [rad/s], \(k_r = 0.12\), \(k_p = 13.1\), \(\omega_f = 5.0 \times 10^3\) [rad/s], \(\omega_d = 213.6\) [rad/s], \(\omega_t = 1.2 \times 10^3\) [rad/s], \(\omega_i = 50.3\) [rad/s], and \(\gamma = -0.3\).

The{frequency response} plots of the first-order harmonics for the PCI-PID and PI\(^2\)D control systems, within the frequency range of \([1, 1000]\) Hz, are presented in Fig. \ref{fig: open_loop_bode_plot_shaped_PCI_PID_final}. Compared to the PI\(^2\)D controller, the PCI-PID controller maintains the same phase margin at the bandwidth frequency of 80 Hz but achieves a higher gain at frequencies lower than 80 Hz and a lower gain at frequencies higher than 80 Hz.
\begin{figure}[htp]
	\centerline{\includegraphics[width=0.9\columnwidth]{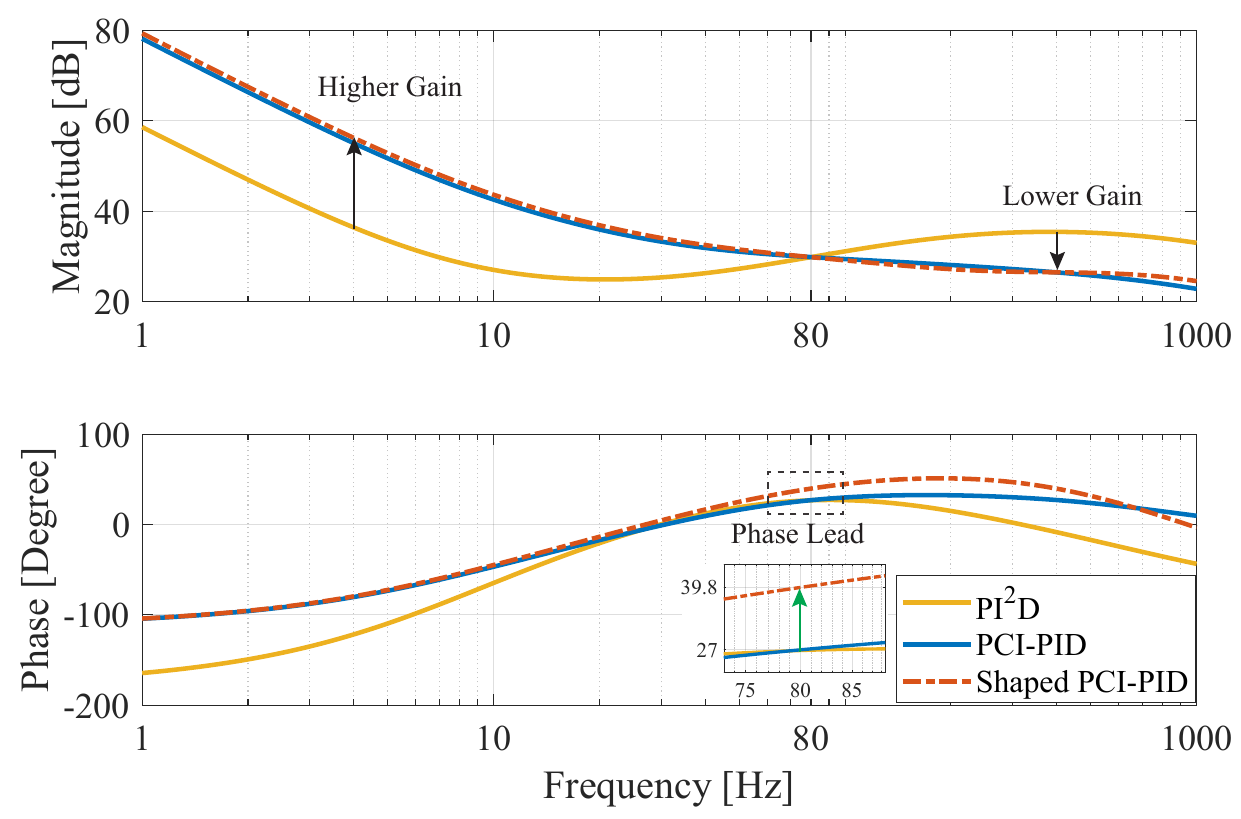}}
	\caption{Bode plots of the first-order transfer functions $\mathcal{L}_1(\omega)$ of open-loop linear PI$^2$D, PCI-PID, and shaped PCI-PID controllers. From here on, black arrows in this study indicate the improvement of reset control over linear control, while green arrows represent the enhancement of shaped reset control compared to reset control.}
	\label{fig: open_loop_bode_plot_shaped_PCI_PID_final}
\end{figure}

By designing a shaping filter for the PCI-PID control system, the objective is to achieve phase lead while controlling gain variations. Setting \(\sigma = 0.1\) limits the gain variation. According to Theorem \ref{thm: ci_Cs_design}, the phase bounds for \(\angle \mathcal{C}_s(\omega)\) are chosen as follows:
\begin{equation}
\label{eq: cs_ineq_ci_ex}
\begin{cases}
\angle \mathcal{C}_s(\omega_c) \in(0,67.08\degree),&\text{ for } \omega=\omega_c,\\
\angle \mathcal{C}_s(\omega)\in \eta_1 = (-25.84^\circ, 25.84^\circ), &\text{ for } \omega \neq\omega_c,
\end{cases}
\end{equation}
The constraint for \(\angle \mathcal{C}_s(\omega)\) where $\omega \neq\omega_c$ in \eqref{eq: cs_ineq_ci_ex} are depicted by the shaded green region in Fig. \ref{fig: cs_plot_ci}. To achieve the desired phase lead relative to the CI, a shaping filter \(\mathcal{C}_s(s)\) is implemented. The transfer function of \(\mathcal{C}_s(s)\) is designed as:
\begin{equation}
\label{eq: cs_design_ci}
    \mathcal{C}_s(s) = \frac{s/950+1}{s/3000 +1} \cdot \frac{1}{s/10^4+1}.
\end{equation}
{Noted that alternative designs of $\mathcal{C}_s(s)$ satisfying the conditions in \eqref{eq: cs_ineq_ci_ex} are feasible. The presented design serves as an example to demonstrate the effectiveness of the shaped reset control design.} 
\begin{figure}[htp]
	\centerline{\includegraphics[width=0.9\columnwidth]{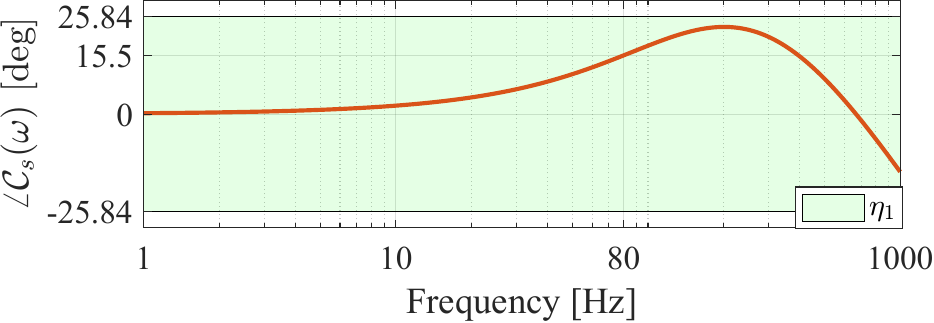}}
	\caption{Plot of \(\angle \mathcal{C}_s(\omega)\) and its bound for the shaped PCI-PID control system.}
	\label{fig: cs_plot_ci}
\end{figure}

As shown in Fig. \ref{fig: cs_plot_ci}, the shaping filter defined in \eqref{eq: cs_design_ci} introduces a phase of \(15.5^\circ\) at the bandwidth frequency of 80 Hz. Since the PCI-PID control system is built upon the CI, the phase lead introduced by the shaping filter is initially applied to the CI and subsequently influences the entire PCI-PID control system. The Bode plots of the CI and the shaped CI, both with \(\gamma = -0.3\), are presented in Fig. \ref{fig: shaped_ci_final}. The shaped CI maintains a gain profile similar to the CI while introducing a phase lead at frequencies below \(665\) Hz, as indicated by the green-shaded region. Specifically, at the bandwidth frequency of 80 Hz, the shaped CI achieves a phase margin of \( -10.1^\circ \), providing a \( 12.8^\circ \) phase lead compared to the \( -22.9^\circ \) phase margin of the CI.
\begin{figure}[htp]
	\centerline{\includegraphics[width=0.9\columnwidth]{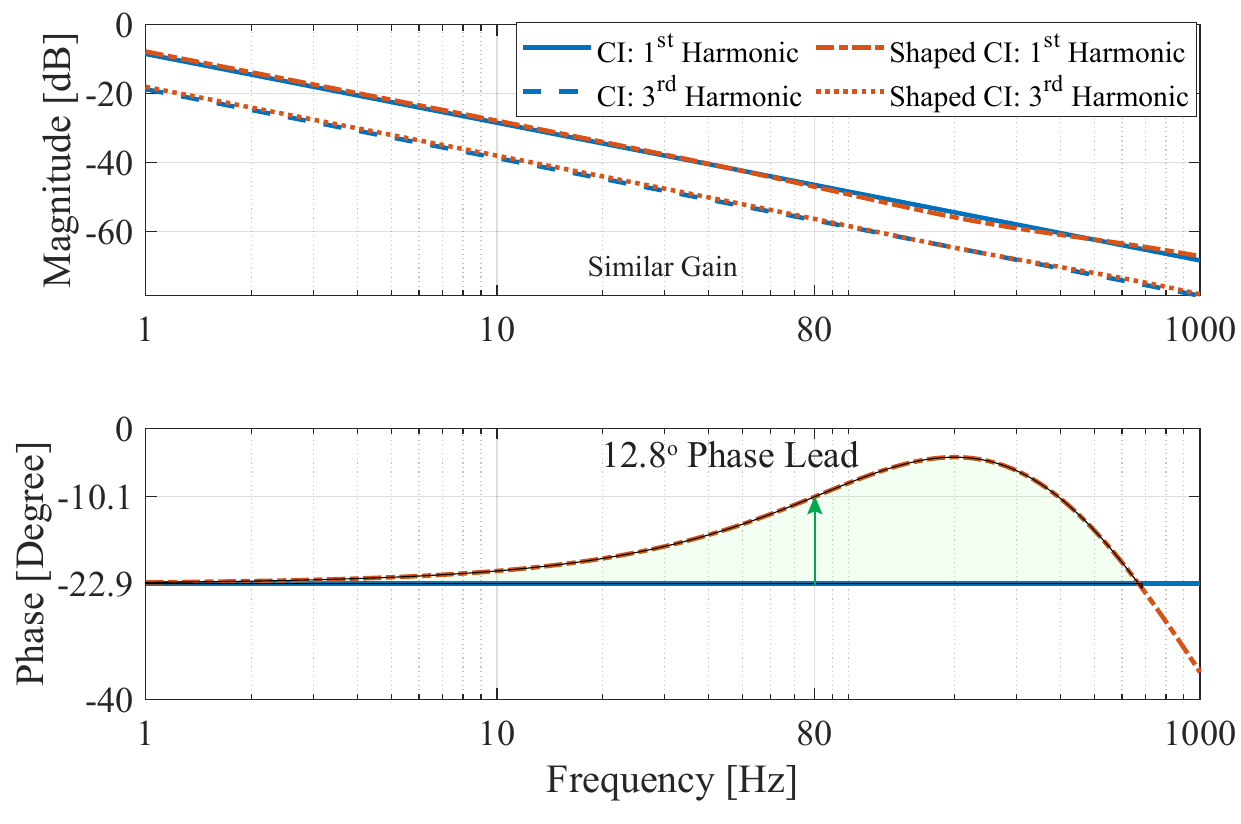}}
	\caption{Bode plots of the CI and the shaped CI with the shaping filter $\mathcal{C}_s$ in \eqref{eq: cs_design_ci}, where $\gamma=-0.3$.}
	\label{fig: shaped_ci_final}
\end{figure}

This designed shaped CI in Fig. \ref{fig: shaped_ci_final} is incorporated into the PCI-PID control system to form the shaped PCI-PID control structure in Fig. \ref{fig: PCI_PID_Structure}. In this configuration, the parameter \(k_r = 0.13\) is adjusted to ensure the same gain as the PCI-PID control system at the 80 Hz bandwidth frequency. As shown in Fig. \ref{fig: open_loop_bode_plot_shaped_PCI_PID_final}, the open-loop Bode plot of the shaped PCI-PID controller closely matches the gain profile of the PCI-PID system but provides a phase lead of \(12.8^\circ\).


Figure \ref{fig: open_loop_bode_plot_shaped_PCI_PID_sys_final} displays the Bode plots for the PI\(^2\)D, PCI-PID, and shaped PCI-PID control systems, implemented on the stage shown in Fig. \ref{fig: Spider_stage}, including both the first- and third-order harmonics. All three systems share the same bandwidth frequency of 80 Hz. Compared to the PI\(^2\)D system, the PCI-PID system maintains the same phase margin of \(27.2^\circ\) but demonstrates higher gain at low frequencies and lower gain at high frequencies. The shaped PCI-PID system behaves even better. It retains similar gain characteristics as the PCI-PID system but achieves a phase margin of \(40^\circ\), with an increased phase margin of \(12.8^\circ\) in the time domain. This \(12.8^\circ\) phase lead is expected to improve the transient response of the system, a benefit that will be validated through experiments.
\begin{figure}[htp]
	\centerline{\includegraphics[width=0.9\columnwidth]{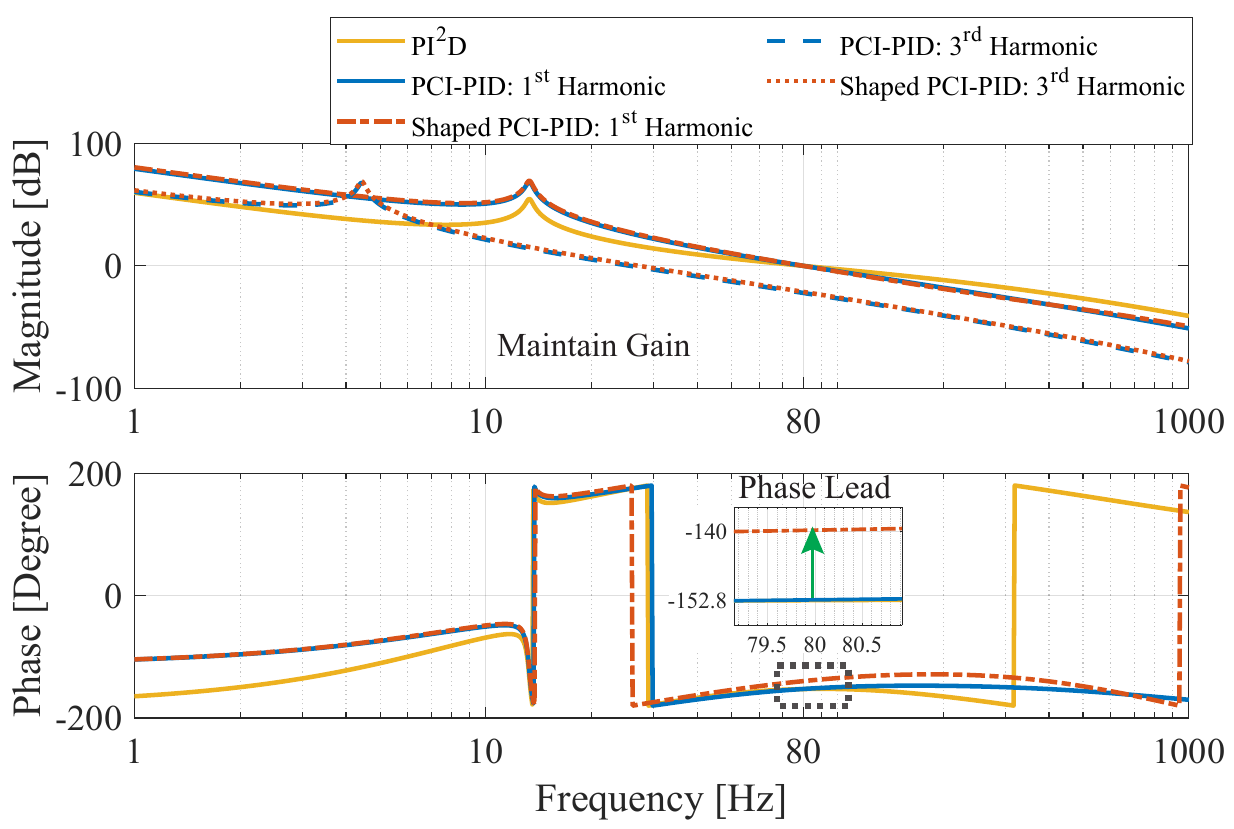}}
	\caption{Bode plots of PI$^2$D, PCI-PID, and shaped PCI-PID control systems. The third-order harmonics of PCI-PID and shaped PCI-PID control systems are shown in dashed lines.}
	\label{fig: open_loop_bode_plot_shaped_PCI_PID_sys_final}
\end{figure}

Figure \ref{fig: EXP_Step_responses_final} illustrates the experimentally measured step responses for the PI\(^2\)D, PCI-PID, and shaped PCI-PID control systems. The overshoot of the PI\(^2\)D and PCI-PID control systems are 64\% and 36\%, respectively, while the shaped PCI-PID achieve the zero overshoot performance. These results highlight the improved transient performance achieved with the shaped reset control, directly attributed to the enhancement in phase lead.
\begin{figure}[htp]
	\centerline{\includegraphics[width=0.9\columnwidth]{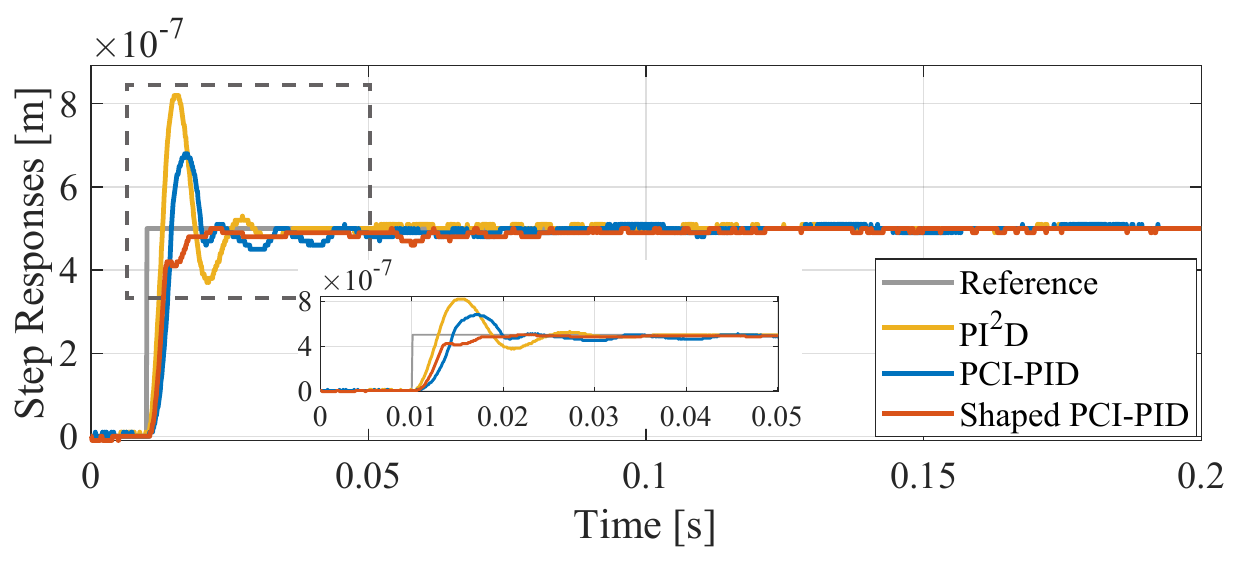}}
	\caption{Experimentally measured step responses of the PI\(^2\)D, PCI-PID, and shaped PCI-PID control systems.}
	\label{fig: EXP_Step_responses_final}
\end{figure}

\subsection{Case Study 2: Gain Benefit of Shaped Reset Control Leading to Steady-State Performance Improvement}
In Case Study 2, a reset CgLp-PID control system is designed to demonstrate the gain benefits of shaped reset control within the generalized FORE-based reset control when \(\omega_\alpha > 0\). The design follows Theorem \ref{thm: fore_Cs_design}.

The CgLp reset element consists of a FORE combined with a lead element, as shown in Fig. \ref{fig: CgLp_PID_structure}. {The transfer function of the PID controller is expressed as}
\begin{equation}
    \text{PID} = k_p \cdot \frac{s + \omega_i}{s} \cdot \frac{s/\omega_d + 1}{s/\omega_t + 1},
\end{equation}
{incorporating a Low-Pass Filter (LPF) given by}
\begin{equation}
\text{LPF} = \frac{1}{s/\omega_f + 1}.
\end{equation}
\begin{figure}[htp]
	\centerline{\includegraphics[width=0.9\columnwidth]{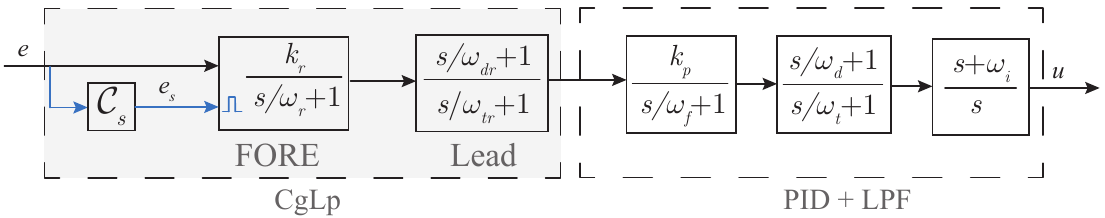}}
	\caption{Block diagram of the CgLp-PID control system.}
	\label{fig: CgLp_PID_structure}
\end{figure}

Compared to a linear PID controller, the CgLp-PID can maintain the same phase lead while benefiting from improved gain (\cite{saikumar2019constant}), as illustrated below. The design parameters for the CgLp-PID controller are: \(\omega_{r} = 160.2\) [rad/s], \(k_r = 1\), \(k_p = 6.5\), \(\omega_{dr} = 336.8\) [rad/s], \(\omega_{tr} = 3.14\times10^4\) [rad/s], \(\omega_{f} = 3.1\times10^3\) [rad/s], \(\omega_d = 143.9\) [rad/s], \(\omega_t = 685.6\) [rad/s], \(\omega_i = 31.4\) [rad/s], and \(\gamma = -0.3\). The design parameters for the PID controller are: \(k_p = 3.0\), \(\omega_d = 81.9\) [rad/s], \(\omega_t = 1.2\times10^3\) [rad/s], \(\omega_f = 3.1\times10^3\) [rad/s], and \(\omega_i = 31.4\) [rad/s]. 

Figure \ref{fig: open_loop_bode_plot_shaped_PCI_PID_final} shows the {frequency response} plots of the first-order harmonic for these systems within the frequency range of \([1, 1000]\) Hz. The CgLp-PID matches the PID in both gain and phase at the bandwidth frequency 50 Hz, while exhibiting higher gain at frequencies lower than 50 Hz and lower gain at frequencies higher than 50 Hz. The following content designs a shaped CgLp-PID controller that maintains the same phase and high-frequency gain properties as the CgLp-PID system while providing improved low-frequency gain and bandwidth benefits.

\begin{figure}[htp]
	\centerline{\includegraphics[width=0.9\columnwidth]{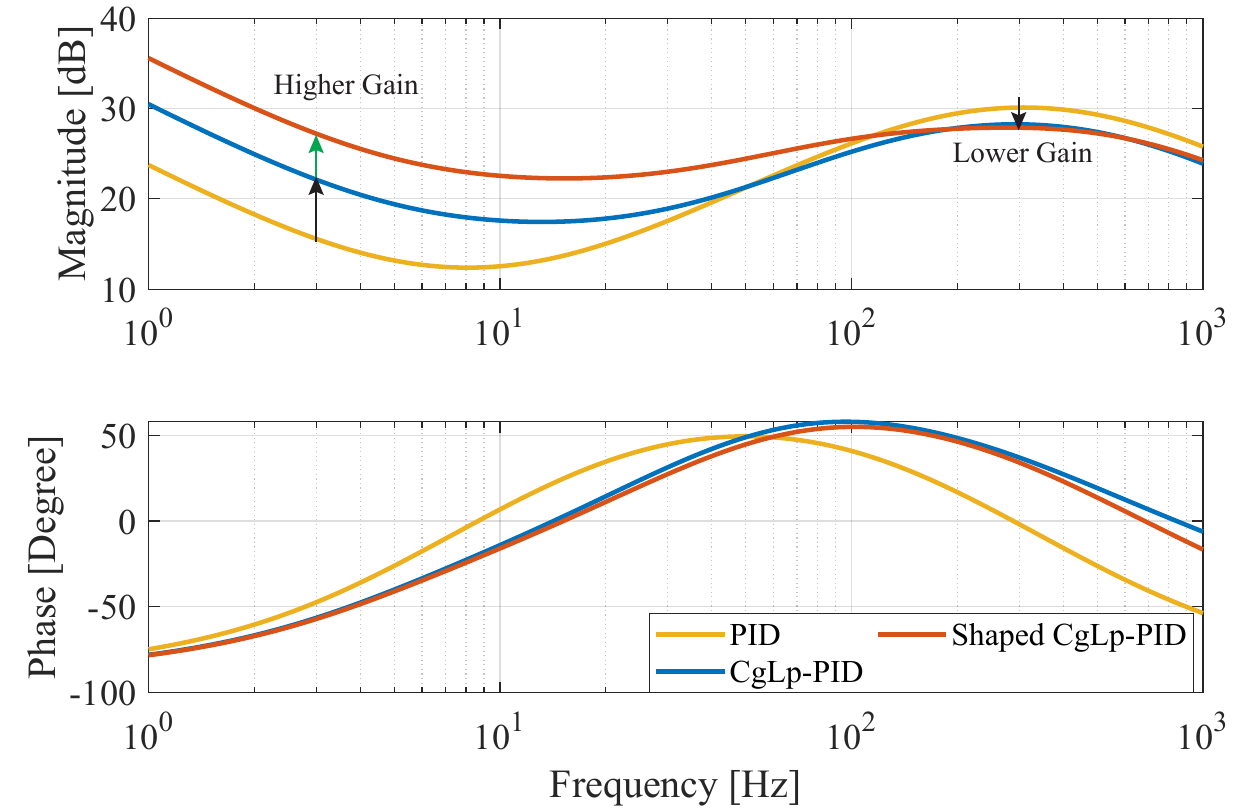}}
	\caption{Bode plots of the first-order transfer functions \(\mathcal{L}_1(\omega)\) for the open-loop linear PID, CgLp-PID, and shaped CgLp-PID controllers.}
	\label{fig: open_loop_plot_shaped_cglp_final}
\end{figure}

The CgLp-PID control system is built upon the FORE. To design a shaped FORE with phase lead, according to Theorem \ref{thm: fore_Cs_design}, by choosing $\sigma = 0.1$, the bound of $\angle \mathcal{C}_s(\omega)$ is chosen as
\begin{equation}
\label{eq: cs_ineq_fore_ex}
\begin{cases}
\angle \mathcal{C}_s(\omega_c) \in(0,27.02\degree),&\text{ for } \omega=\omega_c,\\
\angle \mathcal{C}_s(\omega)\in \beta_1\cup\beta_4, &\text{ for } \omega \neq\omega_c,
\end{cases}
\end{equation}
where
\begin{equation}
\begin{aligned}
\beta_1 &=  (\arctan \theta_\alpha - \arccos  (\theta_\gamma), \arctan \theta_\alpha - \arccos (\theta_\eta)),\\
\beta_4 &=  (\arctan \theta_\alpha + \arccos (\theta_\eta), \arctan \theta_\alpha + \arccos (\theta_\gamma)),\\
\theta_\alpha &= \frac{\omega_r}{\omega},\ \theta_\gamma= \frac{0.9}{\sqrt{1+\theta_\alpha^2}},\ \theta_\eta =  \frac{1.1}{\sqrt{1 + \theta_\alpha^2}}.
\end{aligned}
\end{equation}

The bound specified in \eqref{eq: cs_ineq_fore_ex} for \(\omega \neq \omega_c\) is depicted in Fig. \ref{fig:cs_plot_fore}. A shaping filter \(\mathcal{C}_s(s)\) that adheres to this bound is designed as follows:
\begin{equation}
\label{eq: Cs_fore}
\mathcal{C}_s(s) = \frac{s/950+1}{s/2000 +1} \cdot \frac{1}{s/10^5+1}.
\end{equation}

As shown in Fig. \ref{fig:cs_plot_fore}, the \( \angle \mathcal{C}_s(\omega) \) is \( 10^\circ \) at the bandwidth frequency of 50 Hz. According to \eqref{eq:phi_lead}, the phase of \( \angle \mathcal{C}_s(\omega_c) = 9.2^\circ \) results in a \( \phi_{\text{lead}} = 5.9^\circ \) phase lead in the shaped FORE, compared to the FORE with \( \angle \mathcal{C}_s=1 \).
\begin{figure}[htp]
	\centerline{\includegraphics[width=0.9\columnwidth]{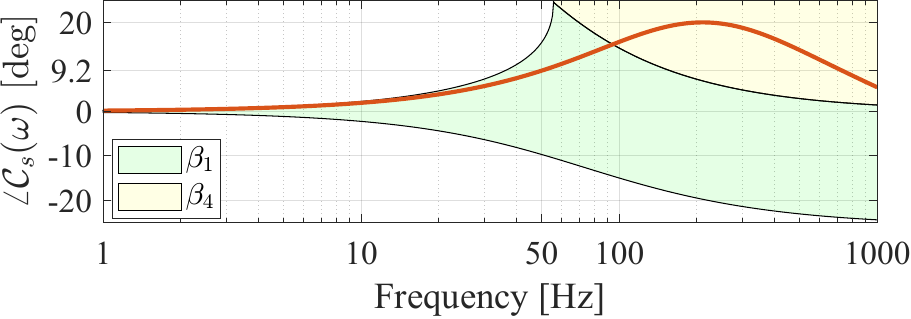}}
	\caption{Plot of $\angle \mathcal{C}_s(\omega)$ and its bounds for the shaped CgLp-PID control system.}
	\label{fig:cs_plot_fore}
\end{figure}

Then, to achieve the desired gain performance while retaining the phase margin, the parameters of the shaped CgLp controller are adjusted to \(\omega_{r}=145.6\) [rad/s], \(k_r = 1.8\), and \(\gamma = 0.08\). The Bode plots of the shaped CgLp-PID control system are presented in Fig. \ref{fig: open_loop_plot_shaped_cglp_final}. 

Then, applying the PID, CgLp-PID, and shaped CgLp-PID controllers to the plant in \eqref{eq:P(s)}, the resulting open-loop Bode plots are presented in Fig. \ref{fig: open_loop_plot_shaped_cglp_sys}. All three systems achieve an identical phase margin of \( 50^\circ \) and similar gain at frequencies higher than 50 Hz. However, the shaped CgLp-PID control system exhibits higher gain than the CgLp-PID at frequencies below 50 Hz. Additionally, the shaped CgLp-PID system achieves a wider bandwidth of 61.6 Hz, compared to 50 Hz for the CgLp-PID system. Although higher-order harmonics show a slight increase at frequencies below 50 Hz, their magnitudes remain negligible relative to the first-order harmonics. The higher gain at low frequencies is expected to enhance precision in that frequency range, which will be further validated through experimental results.
\begin{figure}[htp]
	\centerline{\includegraphics[width=0.9\columnwidth]{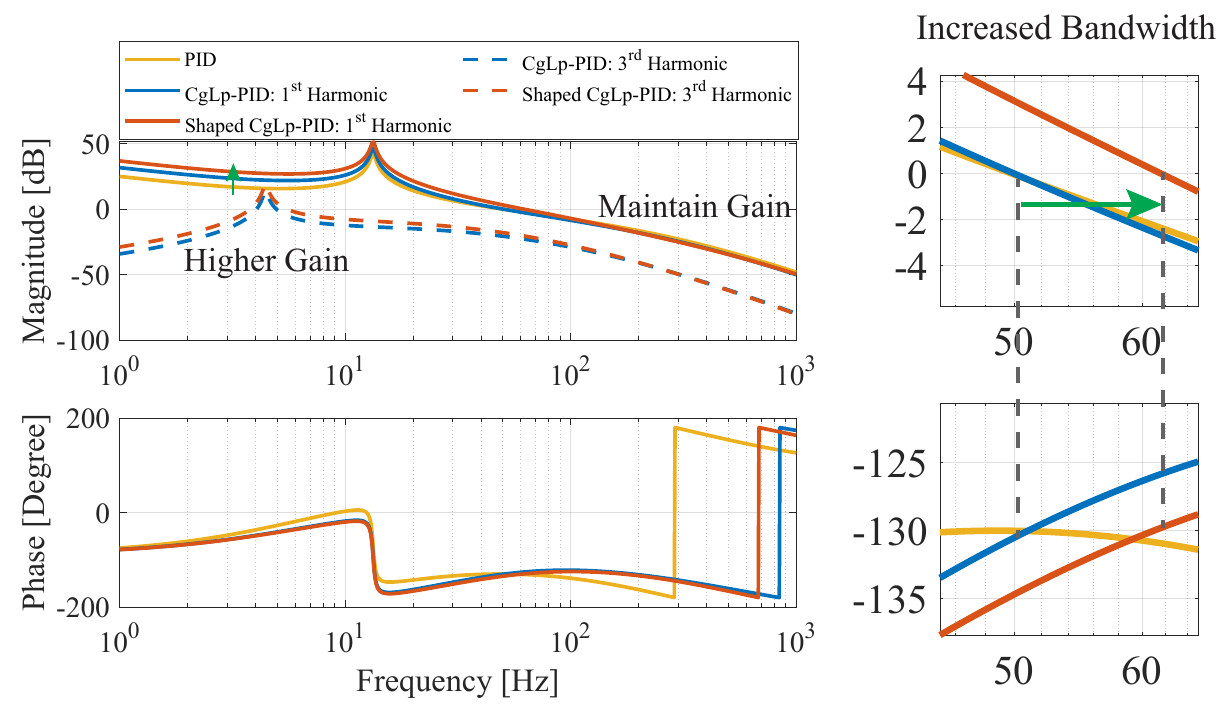}}
	\caption{Bode plots of PID, CgLp-PID, and shaped CgLp-PID control systems. The third-order harmonics of CgLp-PID and shaped CgLp-PID control systems are shown in dashed lines.}
	\label{fig: open_loop_plot_shaped_cglp_sys}
\end{figure}

\subsubsection{Steady-State Performance: Improved Tracking Precision}
As shown in Fig. \ref{fig: open_loop_plot_shaped_cglp_sys}, the shaped CgLp-PID system is designed to have higher gain at frequencies lower than 50 Hz while maintaining similar gain at frequencies higher than 50 Hz. Consequently, to compare the tracking precision of the PID, CgLp-PID, and shaped CgLp-PID control systems, the steady-state errors at input frequencies of 3 Hz, 5 Hz, 10 Hz, and 30 Hz are measured. Additionally, to validate the high-frequency performance is retained, the performance at a input frequency of 200 Hz is also tested.

Figure \ref{fig: exp_rdn_1Hz} presents the measured steady-state errors for the three control systems when tracking a reference signal \(r(t) = 1 \times 10^{-5} \sin(2\pi t)\) [m] at frequencies of 3 Hz, 5 Hz, 10 Hz, 30 Hz, and 200 Hz. The maximum errors \(||e||_\infty\) [m] for each system are summarized in Table \ref{tb: R1_D5_N5}. The results show that the shaped CgLp-PID system achieves a steady-state performance improvement of 41.3\%, 40.0\%, 30.6\%, 25.0\%, and 0 at frequencies of 3 Hz, 5 Hz, 10 Hz, 30 Hz, and 200 Hz, respectively, compared to the CgLp-PID system.
\begin{figure}[htp]
	\centerline{\includegraphics[width=0.9\columnwidth]{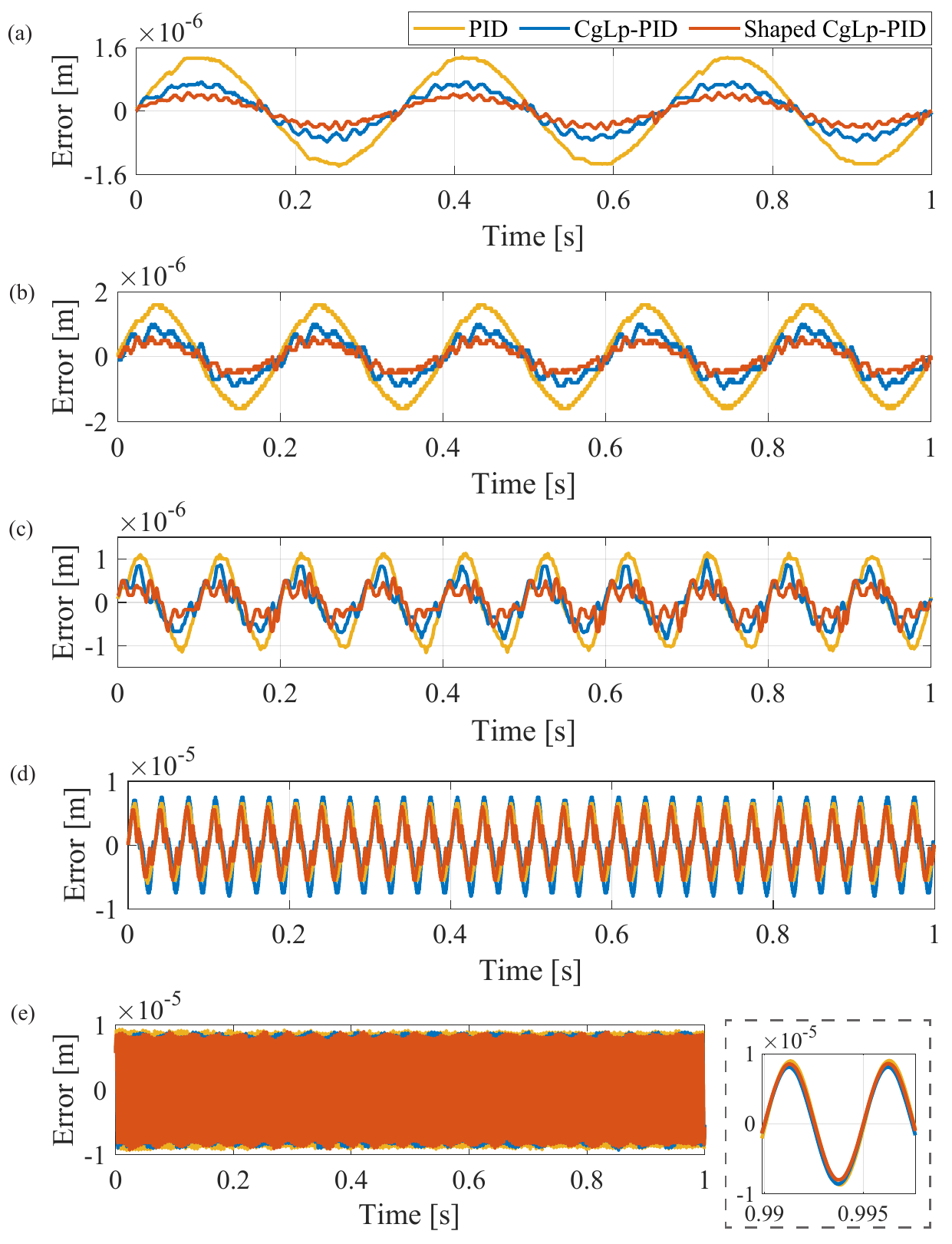}}
	\caption{Experimentally measured steady-state errors of PID, CgLp-PID, and shaped CgLp-PID control systems under reference signals \(r(t) = 1 \times 10^{-5} \sin(2\pi t)\) [m], where \(f =\) (a) 3 Hz, (b) 5 Hz, (c) 10 Hz, (d) 30 Hz, and (e) 200 Hz.}
	\label{fig: exp_rdn_1Hz}
\end{figure}
\begin{table}[htp]
\caption{Maximum steady-state errors \(||e||_\infty\) [m] for the CgLp-PID and shaped CgLp-PID control systems under reference signals \(r(t) = 1 \times 10^{-5} \sin(2\pi t)\) [m], where \(f =\) 3 Hz, 5 Hz, 10 Hz, 30 Hz, and 200 Hz. The precision improvement achieved by the shaped CgLp-PID compared to the CgLp-PID system are highlighted.}
\label{tb: R1_D5_N5}
\centering
    \renewcommand{\arraystretch}{1.5} 
\fontsize{8pt}{8pt}\selectfont
\resizebox{0.9\columnwidth}{!}{
\begin{tabular}{|c|c|c|c|c|c|}
\hline
\multirow{2}{*}{ Systems} & \multicolumn{5}{c|}{Frequency [Hz]}\\ \cline{2-6}
 &3 & 5 & 10 & 30 &200\\ \hline
PID & 1.4$\times 10^{-6}$ & 1.6$\times 10^{-6}$ & 1.2$\times 10^{-6}$& 6.5$\times 10^{-6}$ & 9.4$\times 10^{-6}$ \\ 
CgLp-PID & 8.0$\times 10^{-7}$ & 1.0$\times 10^{-6}$ & 9.8$\times 10^{-7}$& 8.0$\times 10^{-6}$ & 9.3$\times 10^{-6}$\\ 
Shaped CgLp-PID & 4.7$\times 10^{-7}$ & 6.0$\times 10^{-7}$ & 6.8$\times 10^{-7}$& 6.0$\times 10^{-6}$ & 9.3$\times 10^{-6}$\\ 
Precision Improvement & 41.3\% & 40.0\% & 30.6\% & 25.0\% & 0\\ 
\hline
\end{tabular}
}
\end{table}
\subsubsection{Steady-State Performance: Improved Tracking Precision and Disturbance Rejection}
To evaluate the disturbance rejection capability of the shaped CgLp-PID control system, a disturbance signal \(d_1(t) = 1 \times 10^{-8} [75.0\sin(10\pi t) + 7.5 \sin(20\pi t) + 1.5 \sin(40\pi t)]\) [m] is applied to the three control systems. The measured steady-state errors for the PID, CgLp-PID, and shaped CgLp-PID control systems are displayed in Fig. \ref{fig: cglp_d_5_10_15}. The maximum errors for each system are summarized in Table \ref{tb: R1_D5_N5}. The results show that the shaped CgLp-PID system achieves a precision improvement of 40.0\% compared to the CgLp-PID system.
\begin{figure}[htp]
	\centerline{\includegraphics[width=0.9\columnwidth]{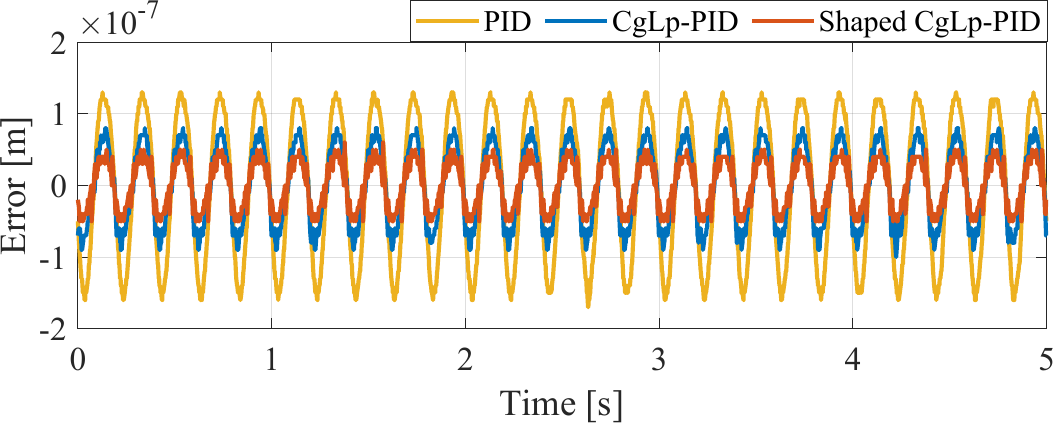}}
	\caption{Experimentally measured steady-state errors of PID, CgLp-PID, and shaped CgLp-PID control systems under a disturbance signal \(d_1(t)\).}
	\label{fig: cglp_d_5_10_15}
\end{figure}
\begin{table}[htp]
\caption{Maximum steady-state errors \(||e||_\infty\) [m] for the CgLp-PID and shaped CgLp-PID control systems under the disturbance signal \(d_1(t)\) and multiple inputs \(r_2(t) + d_2(t)\).}
\label{tb: R1_D5_N5}
\centering
    \renewcommand{\arraystretch}{1.5} 
\fontsize{8pt}{8pt}\selectfont
\resizebox{0.65\columnwidth}{!}{
\begin{tabular}{|c|c|c|}
\hline
\multirow{2}{*}{ Systems} & \multicolumn{2}{c|}{Inputs}\\ \cline{2-3}
& $d_1(t)$ & $r_2(t)+d_2(t)$\\ \hline
PID & 1.7$\times 10^{-7}$ & 1.5$\times 10^{-7}$ \\ 
CgLp-PID & 1.0$\times 10^{-7}$ & 8.0$\times 10^{-8}$ \\ 
Shaped CgLp-PID & 6.0$\times 10^{-8}$ & 5.0$\times 10^{-8}$\\ 
Precision Improvement & 40.0\% & 37.5\% \\ 
\hline
\end{tabular}
}
\end{table}

Then, to assess both reference tracking and disturbance rejection performance, a reference signal \(r_2(t) = 7.5\times 10^{-7}\sin(10\pi t)\) [m] and a disturbance signal \(d_2(t) = 1 \times 10^{-8} [19.1\sin(2\pi t) + 1.8 \sin(4\pi t) + 3.3 \sin(16\pi t)]\) [m] are applied to the three control systems. The measured steady-state errors for the PID, CgLp-PID, and shaped CgLp-PID systems are shown in Fig. \ref{fig: cglp_r5_d128}. The maximum errors for each system are summarized in Table \ref{tb: R1_D5_N5}. The results show that the shaped CgLp-PID system achieves a precision improvement of 37.5\% compared to the CgLp-PID system.
\begin{figure}[htp]
	\centerline{\includegraphics[width=0.9\columnwidth]{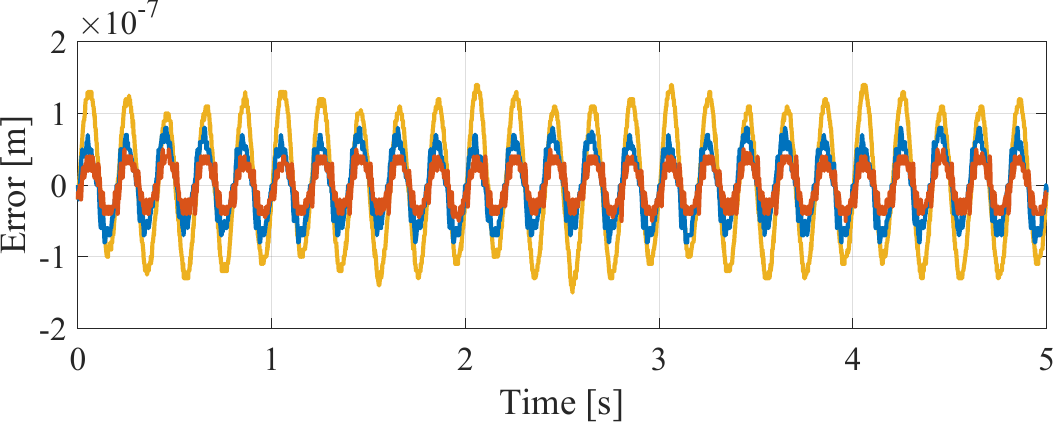}}
	\caption{Experimentally measured steady-state errors for PID, CgLp-PID, and shaped CgLp-PID control systems under multiple inputs: reference signal \(r_2(t)\) and disturbance signal \(d_2(t)\).}
	\label{fig: cglp_r5_d128}
\end{figure}

These results highlight the improved steady-state precision of the shaped CgLp-PID control system, which is attributed to the gain benefits conferred by the shaping filter in the CgLp-PID design, as illustrated in Fig. \ref{fig: open_loop_plot_shaped_cglp_sys}.

\subsubsection{Transient Performance Improvement: Reduced Overshoot}
In addition to enhancing steady-state performance, measurements of the step responses of the three systems, shown in Fig. \ref{fig: exp_cglp_pid_step50_final}, reveal that the shaped CgLp-PID reduces the overshoot observed in the CgLp-PID system, achieving a non-overshoot performance.

This transient performance improvement can be attributed to the introduction of the phase lead element between the error signal \( e(t) \) and the reset-triggered signal \( e_s(t) \), as discussed in the research (\cite{karbasizadeh2022continuous}).

\begin{figure}[htp]
	\centerline{\includegraphics[width=0.9\columnwidth]{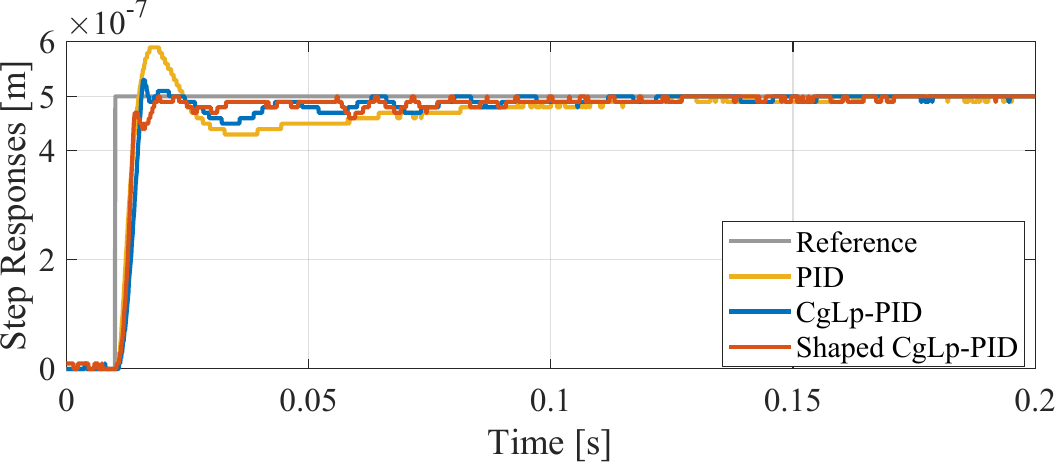}}
	\caption{Experimentally measured step responses of PID, CgLp-PID, and shaped CgLp-PID control systems.}
	\label{fig: exp_cglp_pid_step50_final}
\end{figure}
Thus, the phase-lead shaping filter not only contributes to better steady-state performance but also improves the transient response of the CgLp-PID system.
\section{Conclusion and Discussions}
\label{sec: Conclusion}


In conclusion, this study introduces a phase-lead shaping filter to improve phase and gain characteristics in CI-based and FORE-based reset control systems, referred to as shaped reset control. Frequency-domain design procedures for both CI-based and FORE-based reset control systems are provided. Experimental validation on two reset control systems implemented on a precision motion stage demonstrated the effectiveness of the proposed approach. In the first case study, the shaped reset control enhances transient performance by achieving zero overshoot, benefiting from the phase lead. In the second case study, the shaped reset control improves steady-state precision in reference tracking and disturbance rejection tasks, due to the gain benefit.

However, the benefits of the phase lead shaping filter in \eqref{eq: cs_def} are limited by high-frequency noise in practical systems. The phase lead element can amplify high-frequency noise in the reset-triggered signal, making it necessary to integrate a low-pass filter into the shaping filter. While this low-pass filter mitigates noise amplification, it also reduces some of the benefits provided by the phase lead. When system noise is minimized, the low-pass filter in \eqref{eq: cs_def} can be removed, allowing the advantages of phase lead-shaped reset control to be more pronounced. Future research could explore combining phase lead-shaped reset control with noise reduction techniques, such as the Kalman filter, to further enhance system performance. Investigating the potential of second-order phase lead shaping filters could also provide a promising direction for improvement.

\bibliographystyle{elsarticle-num-names} 

\bibliography{References}
\appendix
\section{Proof of Lemma \ref{lem: phase of Cs to angle C1}}
\label{proof for lema_Cs_C1}
\vspace*{12pt}
\begin{proof}
\label{pf: lemma: angle cs and C1}
This proof derives the condition for the shaping filter \(\mathcal{C}_s\) to increase the phase of the first-order harmonic at the bandwidth frequency, denoted as \(\angle \mathcal{C}_1(\omega_c)\). The proof is divided into two steps: the first addresses the generalized CI when \(\omega_\alpha = 0\), and the second focuses on the FORE when \(\omega_\alpha > 0\).

\vspace{3mm}
\noindent\textbf{Step 1: Condition for the generalized FORE where $\omega_\alpha=0$.}

To ensure that the generalized FORE with a shaping filter $\mathcal{C}_s\neq1$ exhibits a phase lead compared to the system with \(\mathcal{C}_s=1\), we need to ensure:  
\begin{equation}
\label{eq: C_1, C10}
\angle \mathcal{C}_1(\omega_c) > \angle \mathcal{C}_1^0(\omega_c),
\end{equation}
where \(\angle \mathcal{C}_1(\omega_c)\) is the phase of the shaped generalized FORE with the shaping filter \(\mathcal{C}_s(s) \neq 1\), and \(\angle \mathcal{C}_1^0(\omega_c)\) is the phase of the generalized FORE with \(\mathcal{C}_s(s) = 1\).

In the generalized FORE with \(\omega_\alpha = 0\), from \eqref{eq: angle C1}, we have \(\angle \mathcal{C}_1(\omega_c) = \phi_{\lambda}(\omega_c)\). Therefore, to meet the condition in \eqref{eq: C_1, C10}, \(\phi_{\lambda}(\omega_c)\) needs to be larger than its value when \(\mathcal{C}_s(s) = 1\). From \eqref{eq: angle phi_lambda, phi_alpha}, the following condition needs to be satisfied:
\begin{equation}
\label{eq: cond_CI_CS}
   \frac{\sin(2\angle \mathcal{C}_s(\omega_c)) - \pi(1+\gamma)/(2(1-\gamma))}{\cos(2\angle \mathcal{C}_s(\omega_c))+1}  > \frac{-\pi(1+\gamma)}{4(1-\gamma)},
\end{equation}
where the right-hand side corresponds to the element in $\phi_{\lambda}(\omega_c)$ when \(\mathcal{C}_s(s) = 1\).

Then, solving \eqref{eq: cond_CI_CS}, and given the \(\pi\)-period properties of \(\angle \mathcal{C}_s(\omega)\) from Remark \ref{rem: Cs_Cn}, the first condition for the $\angle \mathcal{C}_s(\omega_c)$ in \eqref{eq: CS_Phase_lead_2cond} is derived.

\vspace{3mm}
\noindent\textbf{Step 2: Condition for the generalized FORE where $\omega_\alpha>0$.}

In the generalized FORE with \(\omega_\alpha > 0\), from \eqref{eq: angle C1}, we have  
\[
\angle \mathcal{C}_1(\omega_c) = \phi_{\alpha}(\omega_c) - \arctan\left(\frac{\omega_c}{\omega_\alpha}\right),
\]  
where \(\phi_{\alpha}(\omega_c)\) is an increasing function of \(\kappa_\gamma(\omega_c) \cdot \kappa_\zeta(\omega_c)\), and \(\tan(\angle\mathcal{C}_s(\omega_c))\).

Given the conditions \(\omega > 0\), \(\omega_{\alpha} > 0\), \(\omega_{\beta} > 0\), \(\gamma \in (-1,1)\), and \(\omega > 0\), it follows from the definition of \(\kappa_\gamma(\omega_c)\) in \eqref{eq: angle phi_lambda, phi_alpha} that \(\kappa_\zeta(\omega_c) > 0\). To ensure that the generalized FORE with a shaping filter \(\mathcal{C}_s \neq 1\) achieves a phase lead, both the values of \(\tan(\angle\mathcal{C}_s(\omega_c))\) and \(\kappa_\gamma(\omega_c)\) needs to exceed their respective values in the system where \(\angle \mathcal{C}_s = 0\). This can be achieved by satisfying the following conditions:
\begin{equation}
\label{eq: cond_fore_CS0}
 \angle \mathcal{C}_s(\omega_c) \in (0, k\cdot\pi/2),\ k\in\mathbb{N},
\end{equation}
and
\begin{equation}
\label{eq: cond_fore_CS}
   \omega_c \cdot \cos(2\angle \mathcal{C}_s(\omega_c)) + \omega_\alpha \cdot \sin(2\angle \mathcal{C}_s(\omega_c)) > \omega_c.
\end{equation}
Solving \eqref{eq: cond_fore_CS0} and \eqref{eq: cond_fore_CS}, and given the \(\pi\)-period properties of \(\angle \mathcal{C}_s(\omega)\) from Remark \ref{rem: Cs_Cn}, the second condition for the \(\angle \mathcal{C}_s(\omega_c)\) in \eqref{eq: CS_Phase_lead_2cond} is derived.
\end{proof}

\section{Proof of Lemma \ref{rem: alpha_requirement}}
\label{proof for lema_alpha_requirement}
\vspace*{12pt}
\begin{proof}
This proof establishes the condition required to limit gain changes for a system with a shaping filter compared to a system without the shaping filter at frequencies \(\omega \neq \omega_c\).

From \eqref{eq: fore_cn1} and \eqref{eq: generalized_fore_alpha}, the phase \(\angle \mathcal{C}_s(\omega)\) determines the function \(\alpha(\omega)\), thereby influencing the HOSIDF \(\mathcal{C}_n(\omega)\). The function $\alpha(\omega)$ for the generalized FORE with and without the shaping filter is given by
\begin{equation}
\label{eq: fore_alpha}
    \begin{aligned}
\alpha(\omega) &= 
    \begin{cases}
        \omega, & \text{for } \angle \mathcal{C}_s(\omega) = 0,\\
        e^{j\angle \mathcal{C}_s(\omega)}[\omega \cos(\angle \mathcal{C}_s(\omega))\\
        \indent \indent+\omega_{\alpha}\sin(\angle \mathcal{C}_s(\omega))], &\text{ for } \angle \mathcal{C}_s(\omega)\neq0. 
    \end{cases}
    \end{aligned}
\end{equation}
To limit gain changes of the generalized FORE at frequencies \(\omega \neq \omega_c\), the change in \(\alpha(\omega)\) should be minimized. To evaluate the change in \(\alpha(\omega)\), the ratio of \(\alpha(\omega)\) for the generalized FORE with and without the shaping filter in \eqref{eq: fore_alpha} is defined as:
\begin{equation}
\label{eq: alpha_delta}
 \Delta_\alpha(\omega) = e^{j\angle \mathcal{C}_s(\omega)}[\cos(\angle \mathcal{C}_s(\omega))+\omega_{\alpha}/\omega \sin(\angle \mathcal{C}_s(\omega))] .
\end{equation}
When \( \Delta_\alpha(\omega) \to 1 \) at frequencies \(\omega \neq \omega_c\), the gain properties of the generalized FORE tend to remain unchanged.

From \eqref{eq: alpha_delta}, \( \Delta_\alpha(\omega) \) consists of two components: the phase \(\angle \Delta_\alpha(\omega) = \angle \mathcal{C}_s(\omega)\) and the magnitude given by
\begin{equation}
\label{eq:kappa_alapha2}
\kappa_\alpha(\omega) = |\Delta_\alpha(\omega)| = \left|\cos(\angle \mathcal{C}_s(\omega)) + \frac{\omega_{\alpha}}{\omega} \sin(\angle \mathcal{C}_s(\omega))\right|.
\end{equation}
To ensure that \( \Delta_\alpha(\omega) \) approaches 1, two requirements must be met: First, the phase \(\angle \Delta_\alpha(\omega) = \angle \mathcal{C}_s(\omega)\) should tend to 0. Based on Remark \ref{rem: Cs_Cn}, \(\angle \mathcal{C}_s(\omega)\) affects \(\mathcal{C}_n(\omega)\) with a period of \(\pi\), so \(\angle \mathcal{C}_s(\omega) \to k \cdot \pi\), where \(k \in \mathbb{Z}\) is required. Second, the magnitude \(\kappa_\alpha(\omega)\) should tend to 1.

The constraint \( \kappa_\alpha(\omega) \in (1 - \sigma, 1 + \sigma) \), where \( \sigma \in (0,1) \subset \mathbb{R} \), ensures that both the phase and gain conditions are satisfied. Additionally, as \( \sigma \to 0 \), the change in \( |\mathcal{C}_n(\omega)| \) tends to 0. This concludes the proof.
\end{proof}

\section{Proof of Theorem \ref{thm: ci_Cs_design}}
\label{proof for theorem ci}
\vspace*{12pt}
\begin{proof}
This proof derives the conditions for \(\angle \mathcal{C}_s(\omega)\) in the generalized CI where $\omega_\alpha=0$ to meet the requirements specified in Lemmas \ref{lem: phase of Cs to angle C1} and \ref{rem: alpha_requirement}.

In the generalized CI with \(\omega_\alpha = 0\), from Lemma \ref{lem: phase of Cs to angle C1}, the restriction on \(\angle \mathcal{C}_s(\omega) \in (-\pi, \pi]\) at \(\omega_c\) requires that \(\angle \mathcal{C}_s(\omega_c)\) lies within the bounds: 
\[
\angle \mathcal{C}_s(\omega_c)\in\left(k\pi,\ \frac{\pi}{2}-\arctan\bigg(\frac{\pi(1+\gamma)}{4(1-\gamma)}\bigg)+k\pi\right),\quad k=-1,0.
\]  
From \eqref{eq: kappa_alpha}, the value of \(\kappa_\alpha(\omega)\) is given by:  
\begin{equation}
\label{eq: alpha/w_ci}
\kappa_\alpha(\omega) = |\cos(\angle \mathcal{C}_s(\omega))|.
 \end{equation}
From Lemma \ref{rem: alpha_requirement} and \eqref{eq: alpha/w_ci}, at frequencies where \(\omega\neq\omega_c\), the following condition needs to be satisfied: 
\begin{equation}
\label{eq: ci_low_freqs}
   (1 - \sigma) < |\cos(\angle \mathcal{C}_s(\omega))| < (1 + \sigma), \quad \text{for } \omega \neq \omega_c.
\end{equation}
Given the inherent property of \( \cos(\angle \mathcal{C}_s(\omega)) \in [-1, 1] \) and $\sigma>0$, the condition from \eqref{eq: ci_low_freqs} is expressed as:
\begin{equation}
\label{eq: ci_low_freqs2}
\begin{aligned}
(1 - \sigma) &< \cos(\angle \mathcal{C}_s(\omega)) \leq 1, \text{ or}\\
 -1&\leq \cos(\angle \mathcal{C}_s(\omega)) \leq -1 + \sigma,   \text{ for } \omega \neq \omega_c.
\end{aligned}
\end{equation}
Solving \eqref{eq: ci_low_freqs2}, the conditions for \( \angle \mathcal{C}_s(\omega)\in(-\pi,\pi] \) are given by
\begin{equation}
\label{eq: angleCs_proof}
\begin{aligned}
\angle \mathcal{C}_s(\omega) \in &(-\arccos(1 - \sigma), \arccos(1 - \sigma)) \\
&\cup (\arccos(-1 + \sigma), \pi] \\
&\cup [-\pi, -\arccos(-1 + \sigma)), \text{ for } \omega \neq \omega_c.
\end{aligned}
\end{equation}
Defining \(\eta_1\), \(\eta_2\), and \(\eta_3\) as in \eqref{eq: eta123} and substituting them into \eqref{eq: angleCs_proof} concludes the proof.
\end{proof}

\section{Proof of Theorem \ref{thm: fore_Cs_design}}
\label{proof for theorem fore}
\vspace*{12pt}
\begin{proof}
This proof derives the conditions for \(\angle \mathcal{C}_s(\omega)\in(-\pi,\pi]\) in the FORE where $\omega_\alpha>0$ to meet the requirements specified in Lemmas \ref{lem: phase of Cs to angle C1} and \ref{rem: alpha_requirement}.

From Lemma \ref{lem: phase of Cs to angle C1}, at frequencies where \(\omega=\omega_c\), the following condition needs to be satisfied: 
\begin{equation}
\label{eq:fore_wc}
 \angle \mathcal{C}_s(\omega_c) \in\left(k\pi,\ \frac{\pi}{2}-\arctan \bigg(\frac{\omega_c}{\omega_\alpha}\bigg)+k\pi\right),\quad k=-1,0.
\end{equation}
From \eqref{eq: kappa_alpha}, the function $\kappa_\alpha(\omega) $ can be written as
\begin{equation}
\label{eq: |alpha/w|}
\begin{aligned}
\kappa_\alpha(\omega) &= |\cos(\angle \mathcal{C}_s(\omega)) + \frac{\omega_{\alpha} }{\omega}\sin(\angle \mathcal{C}_s(\omega))|\\
&= \sqrt{1 + \theta_\alpha ^2} \bigg|\cos(\angle \mathcal{C}_s(\omega) -\arctan  \theta_\alpha)\bigg|,
\end{aligned}
\end{equation}
where
\begin{equation}
    \begin{aligned}
\theta_\alpha &= \frac{\omega_\alpha}{\omega}.
    \end{aligned}
\end{equation}
From Lemma \ref{rem: alpha_requirement}, at $\omega\neq\omega_c$, the following condition needs to be satisfied:
\begin{equation}
\label{eq: cond_|alpha_w|}
(1 - \sigma) < \kappa_\alpha(\omega)< (1 + \sigma),   \text{ for } \omega \neq \omega_c.
\end{equation}
From \eqref{eq: |alpha/w|} and \eqref{eq: cond_|alpha_w|}, at $\omega\neq\omega_c$, the following condition needs to be satisfied:
\begin{equation}
\label{eq: cond_alpha_w1}
\begin{aligned}
&0<  \frac{(1 - \sigma)}{\sqrt{1 + \theta_\alpha ^2}  }<   \cos(\angle \mathcal{C}_s(\omega) -\arctan  \theta_\alpha)   < \frac{(1 + \sigma)}{\sqrt{1 + \theta_\alpha ^2}  }, \text{ or }\\
&\frac{(-1 - \sigma)}{\sqrt{1 + \theta_\alpha ^2}  } <  \cos(\angle \mathcal{C}_s(\omega) -\arctan  \theta_\alpha)   < \frac{(-1 + \sigma)}{\sqrt{1 + \theta_\alpha ^2}  } <0,
\end{aligned}
\end{equation}
Solving \eqref{eq: cond_alpha_w1}, the resulting conditions for \(\angle \mathcal{C}_s(\omega)\) are given in \eqref{eq: cs_ineq_fore}. Note that \(\arccos(x)\) is defined within the interval \([0, \pi]\). This completes the proof.

\end{proof}


\end{document}